\documentclass[pra, singlecolumn, superscriptaddress]{revtex4}
\usepackage[latin9]{inputenc}
\usepackage{amsmath}
\usepackage{amssymb, enumerate}
\usepackage{color}
\usepackage{babel}
\usepackage{graphicx}
\usepackage{epstopdf}
\usepackage{float}
\usepackage{thmtools}
\usepackage{thm-restate}
\usepackage{hyperref}
\usepackage{cleveref}
\usepackage{enumerate}
\usepackage{lipsum}
\usepackage{fixmath}
\usepackage{dsfont}

\usepackage{amsthm}
\usepackage{amsmath}
\usepackage{amssymb}
\usepackage{tikz}
\usepackage{color}
\usetikzlibrary{matrix}
\usepackage{dsfont}
\usepackage{graphicx} 
\usepackage{stackrel}
\makeatletter

\@ifundefined{textcolor}{}
{%
 \definecolor{BLACK}{gray}{0}
 \definecolor{WHITE}{gray}{1}
 \definecolor{RED}{rgb}{1,0,0}
 \definecolor{GREEN}{rgb}{0,1,0}
 \definecolor{BLUE}{rgb}{0,0,1}
 \definecolor{CYAN}{cmyk}{1,0,0,0}
 \definecolor{MAGENTA}{cmyk}{0,1,0,0}
 \definecolor{YELLOW}{cmyk}{0,0,1,0}
}

\newenvironment{protocol*}[1]
  {
    \begin{center}
      \hrulefill\\
      \textbf{#1}
  }
  {
    \vspace{-1\baselineskip}
    \hrulefill
    \end{center}
  }
  
\theoremstyle{plain}
\newtheorem{thm}{Theorem}[section]

\theoremstyle{definition}

\theoremstyle{remark}

\newtheorem{lemma}{Lemma}  
\def\bel{\begin{lemma}}
\def\eel{\end{lemma}}
\newtheorem{theorem}{Theorem}

\newtheorem{proposition}[theorem]{Proposition}

\newtheorem{definition}[theorem]{Definition}

\makeatother

\begin{document}

\title{No-signaling-proof randomness extraction from public weak sources}
\author{Ravishankar Ramanathan}
\affiliation{Department of Computer Science, The University of Hong Kong, Pokfulam Road, Hong Kong}
\author{Micha{\l} Banacki}
\affiliation{International Centre for Theory of Quantum Technologies, University of Gda\'{n}sk, Wita Stwosza 63, 80-308 Gda\'{n}sk, Poland}
\affiliation{Institute of Theoretical Physics and Astrophysics, National Quantum Information Centre, Faculty of Mathematics, Physics and Informatics, University of Gda\'{n}sk, Wita Stwosza 57, 80-308 Gda\'{n}sk, Poland}
\author{Pawe{\l} Horodecki}
\affiliation{International Centre for Theory of Quantum Technologies, University of Gda\'{n}sk, Wita Stwosza 63, 80-308 Gda\'{n}sk, Poland}
\affiliation{Faculty of Applied Physics and Mathematics, National Quantum Information Centre, Gda\'{n}sk University of Technology, Gabriela Narutowicza 11/12, 80-233 Gda\'{n}sk, Poland}


\begin{abstract}
The extraction of randomness from weakly random seeds is a topic of central importance in cryptography. Weak sources of randomness can be considered to be either private or public, where public sources such as the NIST randomness beacon broadcast the random bits once they are generated. The problem of device-independent randomness extraction from weak public sources against no-signalling adversaries has remained open. In this paper, we show protocols for device-independent and one-sided device-independent amplification of randomness from weak public sources that use a finite number of devices and are secure against no-signaling adversaries. Specifically, under the assumption that the device behavior is as prescribed by quantum mechanics the protocols allow for amplification of public $\epsilon$-SV sources for arbitrary initial $\epsilon \in [0,0.5)$. On the other hand, when only the assumption of no-signaling between the components of the device is made, the protocols allow for amplification of a limited set of public weak sources. 
\end{abstract}



\maketitle

\section{Introduction.}
The extraction of randomness from weak seeds is a topic of great importance given the utility of uniformly random bits in multiple cryptographic primitives. In most of these applications, one requires private bits that are uniform even given any side information possessed by a malicious adversary. Many physical sources of randomness, such as from radioactive decay or thermal noise, may be considered to be weak in that they produce biased correlated bits. Furthermore, many sources of randomness are also public in the sense that the partially random bits are available to everyone once they are generated. A paradigmatic example of a weak public source is the NIST randomness beacon \cite{nist}, which produces random numbers that are publicly visible to everyone on the internet. In particular, such public random bits can be utilized in statistical sampling, Monte Carlo simulations and fair resource sharing, while private random bits are a more precious resource, useful in establishing shared keys and initiating authenticated connections.  

It is well known that it is impossible, given a single weak source of randomness and classical resources, to create close to uniform random bits \cite{SV84}. A paradigmatic model for the source of randomness is that of the $\epsilon$-SV source \cite{SV84}, a model of a biased coin where the individual coin tosses are not independent but rather the bits $R_i$ produced by the source obey 
\begin{eqnarray}
\frac{1}{2} - \epsilon \leq P\left(R_i = 0 | R_{i-1}, \dots, R_1 \right) \leq \frac{1}{2} + \epsilon.
\end{eqnarray}
The parameter $0 \leq \epsilon < \frac{1}{2}$ describes the reliability of the source, and the task can be expressed as being to convert a single $\epsilon$-SV source with $\epsilon < \frac{1}{2}$ into one with $\epsilon \rightarrow \frac{1}{2}$. No classical deterministic procedure exists for this task, which is termed as \textit{randomness amplification} or \textit{randomness extraction} \cite{CR12, CSW14}. In particular, classical randomness extractors either require an additional independent random seed or multiple independent weak sources in order to produce uniformly random bits \cite{CG88}.

The discovery that quantum resources can help in achieving the task provided a major fillip to the field of quantum cryptography \cite{CR12}. Furthermore, the task is considered within the paradigm of Device-Independent (DI) quantum cryptography, which achieves the highest form of security called information-theoretic security based on the very laws of nature. In this paradigm, users do not even need to trust the devices executing the cryptographic protocol, and can instead verify its correct execution by means of simple statistical tests on the devices. These statistical tests take the form of Bell tests, that verify that the  correlations in the devices satisfy the conditions for quantum non-local correlations. In its highest form, such a cryptographic protocol guarantees security based on the simplest and most fundamental law of nature, namely the rule of no-superluminal signalling of relativity, i.e., the adversary is said to hold no-signaling side information \cite{CR12, GMTD+13, BRGH+16} (for a more general principle based only on relativistic causality, see \cite{HR19}).

Relaxations of the highest security guarantees afforded by no-signaling secure DI protocols have been considered. In particular, one may consider security against adversaries limited to quantum or even classical side information about the devices \cite{CSW14, KAF17, FWEBC20}. Furthermore, it is also possible to consider situations in which the devices used by the honest parties are partially trusted. This trust can take several forms, one may trust the state preparation procedure or measurements performed on the devices \cite{Pirandola2019}. An interesting paradigm is that of \textit{one-sided device-independence} ($1$sDI) where one of the parties (a trusted party such as NIST, say) involved in the protocol is assumed to have full control over their device \cite{BCW+12}. In particular (as considered in the second part of this paper), one of the honest parties holds a well-characterized quantum system upon which they may perform trusted measurements including for quantum tomography. The statistical tests in this paradigm take the form of tests for the violation of steering inequalities, that verify that the correlations in the devices satisfy the criteria for quantum steering. 

Besides the cryptographic importance of private random bits, the task of randomness amplification is also of foundational interest with implications for the philosophical problem of the existence of free-will \cite{CR12, HR19}. The existence of a physical process producing fully uniform bits starting from arbitrarily weak seeds under the assumption of no-signalling can be thought of, in that context, as the statement that given a speed limit to the propagation of information, the existence of an arbitrarily small amount of free-will in nature implies the existence of a complete freedom of choice. Thus for both fundamental and cryptographic reasons, it is of central interest to produce simple, practical and efficient DI and $1$sDI protocols for randomness amplification of weak public sources and show their security on the basis of the no-signalling principle. 

In this paper, we present a device-independent protocol for amplification of a single weak public SV source (with $\epsilon < 0.09$) using a device consisting of two spatially separated components and show a proof of its composable security against no-signaling side information. When the device behavior is as prescribed by quantum theory, we show a much stronger result that the amplification is possible for the entire range of $\epsilon \in [0, 0.5)$. While the protocol is secure against general no-signaling adversaries, it only tolerates a vanishing rate of noise and has vanishing output rate. As such, while it represents a foundational advance, the scheme is not practical at its present stage of development. In the second part of the paper, we introduce a one-sided device-independent scheme for randomness amplification of arbitrarily weak public SV sources secure against no-signaling side information. This latter protocol utilises a device consisting of three spatially separated components (of which one is a trusted qubit system) and works for arbitrary no-signaling adversaries. We show that the protocol yields secure private bits in the noiseless scenario by analysing the extremal points of tripartite time-ordered no-signaling assemblages. We defer the optimal rate calculations for the analogous noise-tolerant version of the protocol for future work. 

\section{Device-Independent Amplification of weak public sources.}
Colbeck and Renner \cite{CR12} were the first to show a DI protocol for the amplification of randomness from SV sources. Their protocol had the important feature of amplifying weak public sources using a device consisting of only two spatially separated components, as well as being secure against general no-signalling adversaries. 
However, it also had a few drawbacks including being applicable only to a restricted set of public SV sources satisfying $\epsilon < 0.058$, requiring a large number of measurements (growing with the security parameter), tolerating a vanishing rate of noise and producing vanishing extraction rate. Nevertheless, their result was also important from a more fundamental point of view. Namely, the result shows that the presence of weakly random process in nature (up to $\epsilon < 0.058$) together with the no-signalling assumption imply the existence of fully random processes. In the first part of the paper, we give an improved version of their protocol that achieves amplification of sources up to $\epsilon \lessapprox 0.090$.



Since \cite{CR12}, advances have been made in multiple works with security proven against both quantum \cite{CSW14, KAF17} and no-signalling adversaries \cite{GMTD+13}. The no-signaling paradigm, constrained by the impossibility of sending messages instantaneously (i.e., of signaling) has an appealing advantage over the paradigm with quantum adversaries. Namely, while knowledge of quantum mechanics is necessary in order to build the devices implementing the protocol, verifying whether the device provides  true randomness does not require any such knowledge. Security of the randomness produced by the device can be verified by honest everyday users using knowledge of basic statistics.

While \cite{GMTD+13} tackled the problem of amplifying weak public sources against  no-signalling side information, it required a large number of spatially separated devices (polynomial in the number of bits taken from the source) making it quite impractical. In \cite{BRGH+16, RBHH+16}, we showed some practical noise-tolerant protocols for amplifying arbitrarily weak SV sources against no-signaling side information. However, these latter protocols were designed to handle private rather than public sources, i.e., they assume that any and all of the bits produced by the source are kept private, and are never leaked to the adversary even after completion of the protocol. It is an open problem to quantitatively analyse how much the security of these protocols is affected when (a portion of the) bits used in the protocol are leaked to the adversary.

The task of public source amplification against no-signaling adversaries faces a major obstacle that has prevented further progress on the problem. Namely, the usual structure of a randomness amplification protocol, of performing a Bell test and applying a multi-source randomness extractor to the outputs (along with any further bits from the weak source), does not seem to work in this scenario. There are significant indications that no-signalling proof multi-source randomness extractors do not exist \cite{RA12}. The task of estimating the min-entropy from the outputs of the device, for which tools such as the \textit{entropy accumulation theorem} have been developed in the setting of quantum side information, also becomes difficult when considering the general scenario of no-signaling side information. In this paper, we  overcome these difficulties with a different proof. We show the following result. 
\begin{theorem}
\label{thm:result-inf}
There exists a protocol using a device consisting of two spatially separated components that takes as inputs bits $R_i$ from an $\epsilon$-SV source (of arbitrary fixed $\epsilon < 0.09$) and outputs a bit $R_r$ such that under the assumption of no-signaling between the components of the device, the bit $R_r$ is certified to be arbitrarily free, except with negligible probability. Under the additional assumption that the device behavior is as predicted by quantum theory, the statement holds for arbitrary initial $\epsilon < 0.5$.
\end{theorem}

\begin{figure}
	\begin{protocol*}{Protocol I}
		\begin{enumerate}
			\item The $\epsilon$-SV source is used to choose the measurement settings $(\textbf{x}_i, \textbf{y}_i)$ (with each $\textbf{x}_i$ and $\textbf{y}_i$ chosen by taking $\log (N+1)$ bits from the source) for $M$ runs on a single device consisting of two components. The device produces output bits $(\textbf{a}_i, \textbf{b}_i) \in \{0,1\}^2$ with $i \in \{1,\dots, M\}$.
			\item Define $\mathcal{S}_{\text{H}}$ to be the subset of measurement runs in which the measurement settings involved in the Hardy ladder test \eqref{eq:Hardy-ineq-1} appear, i.e., $\mathcal{S}_{\text{H}} := \big\{ i \in \{1, \dots, M \}| (\textbf{x}_i, \textbf{y}_i) = (k, k-1) \; \text{for} \; k \in \{1, \dots, N\} \vee (\textbf{x}_i, \textbf{y}_i) = (k-1, k) \; \text{for} \; k \in \{1, \dots, N\} \vee (\textbf{x}_i, \textbf{y}_i) = (0, 0) \vee (\textbf{x}_i, \textbf{y}_i) = (N, N) \big\}$. The parties check that the cardinality of the set $\mathcal{S}_{\text{H}}$ satisfies $|\mathcal{S}_{\text{H}}| \in \left[\frac{M}{(N+1)}, \frac{3M}{(N+1)}\right]$. If not, they set the final output bit $R_r = \perp$ and abort.
			\item Define $\mathcal{S}_{\text{H-nz}}$ to be the subset of measurement runs in which the measurement setting $(N,N)$ appears, i.e., $\mathcal{S}_{\text{H-nz}} := \big\{i \in \{1, \dots, M\} | (\textbf{x}_i, \textbf{y}_i) = (N,N) \big\}$. The parties check that the cardinality of the set $\mathcal{S}_{\text{H-nz}}$ satisfies $\big| \mathcal{S}_{\text{H-nz}} \big| \in \left[\frac{M}{2(N+1)^2}, \frac{3M}{2(N+1)^2} \right]$. If not they set the final output bit $R_r = \perp$ and abort. 
			\item 
			The parties check that the quantity $\textit{Z}_{\text{H}} :=  \frac{1}{|\mathcal{S}_{\text{H}}|}\sum_{i \in \mathcal{S}_{\text{H}}} B_{\text{H}}(\textbf{a}_i, \textbf{b}_i, \textbf{x}_i, \textbf{y}_i) \leq \delta_1$, for fixed $\delta_1 > 0$, i.e., the parties check that the Hardy zero probability events do not occur (other than with some small frequency $\delta_1$). Otherwise, the protocol is aborted and the parties set $R_r = \perp$.
			\item The parties perform a partial tomographic test by computing the empirical average $Z_{\text{H-nz}} := \frac{1}{|\mathcal{S}_{\text{H-nz}}|} \sum_{i \in \mathcal{S}_{\text{H-nz}}} B_{\text{H-nz}}(\textbf{a}_i, \textbf{b}_i, \textbf{x}_i, \textbf{y}_i)$. The protocol is aborted and the parties set $R_r = \perp$ unless $Z_{\text{H-nz}} \geq \frac{1}{2} - \delta_2$ for fixed $\delta_2 > 0$.
			\item Conditioned on not aborting in the previous steps, the parties choose $r \in \mathcal{S}_{\text{H-nz}}$ using $\log | \mathcal{S}_{\text{H-nz}} |$ bits from the $\epsilon$-SV source and output the bit $R_r$ obtained from the measurement outcomes of the $r$-th run as
 \[R_r =  \begin{cases} 
      0 & (\textbf{a}_r, \textbf{b}_r) = (0,0) \\
      1 & (\textbf{a}_r, \textbf{b}_r) \neq (0,0)
   \end{cases}
\]
		\end{enumerate}
	\end{protocol*}
	\caption{\label{protocolsingle} Protocol for device-independent randomness amplification of a public SV source.}
\end{figure}
Theorem \ref{thm:result-inf} is realized by Protocol I shown in Fig. \ref{protocolsingle}. The quantities $B_{\text{H}}(\textbf{a}_i, \textbf{b}_i, \textbf{x}_i, \textbf{y}_i)$ and $B_{\text{H-nz}}(\textbf{a}_i, \textbf{b}_i, \textbf{x}_i, \textbf{y}_i)$ are specific indicator functions for the tests performed in the protocol and are shown together with the tested Bell inequality and the proof of security of the Protocol I in the Supplemental Material. A key ingredient in the protocol is a modified version of Hardy's 'ladder" argument for non-locality, given in Section \ref{sec:Bell-ineq} in the Supplemental Material, for which we derive the optimal quantum value. A second key ingredient in the protocol is the introduction of two statistical tests, performed by the honest parties, which allow us to certify that the final output bit $R_r$ from any no-signaling box that passes the tests must necessarily be arbitrarily free. Particular modifications to the estimation procedure are introduced in the proof here to allow us to reach this conclusion for a weak initial public seed.

While the Protocol I achieves amplification of arbitrarily weak public sources against no-signaling adversaries, it requires a large number of measurement settings (that grows as an inverse polynomial in the distance to uniform of the final output bit) to do so. The Protocol I also cannot tolerate a constant rate of noise (the distance to uniform of the output bit is directly proportional to the degree of quantum violation of the inequality) and has zero extraction rate. Furthermore, the devices are assumed to be fixed given the source (more precisely, this is the assumption that given the history and Eve's side information, the device and the string of bits produced by the source are independent). It is an important question for future work, whether some of these conditions can be relaxed. In particular, a route to relaxing the source-device model introduced in \cite{WBG+17} appears to be a promising direction for further exploration. 


\section{One-sided Device-Independent Amplification of weak public sources.} In the previous section, we have seen a protocol that allows honest users to achieve the goal of extracting randomness from a weak public SV source. However, at its present stage of development, the protocol only affords a vanishing extraction rate as well as vanishing noise-tolerance. In this section, we show a more practical  protocol for the same task, i.e., we show a protocol that allows to extract a linear fraction of the random bits. The price to pay is to assume that one of the parties, Alice say, holds a trusted qubit system on which she is able to perform state preparation and measurements in a fully trusted manner. We thus work within the paradigm of one-sided device-independence (1sDI) \cite{BCW+12}. In this paradigm, the parties test for the violation of an EPR-steering inequality \cite{Sch35, CJWR09} in an analogous manner to the testing of Bell inequality violation in the DI paradigm. In particular, assuming that one of the parties' system has a definite (though unknown) quantum state, the protocol must verify that the other parties, by their choice of measurement on their black box devices, can affect this state. It has been shown \cite{BCW+12} that the experimental requirements (in particular, the detection efficiencies required) for practical implementations of 1sDI protocols are much lower than for the DI protocols, making 1sDI a promising approach for devising protocols feasible with exisiting devices. 

We show a 1sDI randomness amplification protocol II for extracting randomness from arbitrarily weak SV sources secure against general no-signaling adversaries. The protocol shown in Fig. \ref{protocol1sDI} considers a steering scenario with three honest parties (Alice, Bob and Charlie) and a no-signaling adversary Eve. That is, the devices held by the honest parties and adversary Eve are such that one honest user (Charlie) holds a trusted characterised quantum system while the devices held by the other parties and Eve are only constrained by the no-signaling principle. The devices are thus characterised as general no-signaling assemblages \cite{SBCSV15}, . The formulation of 1sDI randomness amplification protocols secure against no-signaling adversaries has been facilitated by our recent results showing that quantum correlations can realise extremal non-classical no-signaling assemblages \cite{RBRH20, BMRH21} in a three-party scenario. Furthermore, the corresponding steering inequalities can be formulated in the framework of arbitrarily limited measurement independence, i.e., for measurements chosen using $\epsilon$-SV sources of arbitrary $\epsilon < 1/2$. 

\begin{figure}
	\begin{protocol*}{Protocol II}
		\begin{enumerate}
			\item The $\epsilon$-SV source is used to choose the measurement settings $(\textbf{x}_i, \textbf{y}_i)$ for two parties Alice and Bob (with each $\textbf{x}_i$, $\textbf{y}_i$ chosen by taking one bit from the source) for $M$ runs on a single device consisting of two components. The device produces output bits $(\textbf{a}_i, \textbf{b}_i) \in \{0,1\}^2$ with $i \in \{1,\dots, M\}$ and the corresponding state $\sigma_{\textbf{a}_i, \textbf{b}_i | \textbf{x}_i, \textbf{y}_i}$ of the third party Charlie.
			\item 
			The parties check that the quantity $\textit{Z}_{\text{S}} :=  \frac{1}{M}\sum_{i =1}^{M} B_{\text{S}}(\textbf{a}_i, \textbf{b}_i, \textbf{x}_i, \textbf{y}_i) =0$, i.e., the parties check for algebraic violation of the steering inequality in all the $M$ runs. Otherwise, the protocol is aborted and the parties set $R_r = \perp$.
			\item Conditioned on not aborting in the previous step, the parties choose $r \in \{1, \dots, M\}$ using $\log M$ bits from the $\epsilon$-SV source and output Alice's measurement outcome $\textbf{a}_r$ from the $r$-th run as the final output bit $R_r$.
		\end{enumerate}
	\end{protocol*}
	\caption{\label{protocol1sDI} Protocol for one-sided device-independent randomness amplification of a public SV source.}
\end{figure}

In the Supplemental Material, we sketch the proof of security of the one-sided device-independent (1sDI) randomness amplification protocol II, namely of the following theorem.
\begin{theorem}
\label{thm:result-inf-2}
There exists a one-sided device-independent protocol using a device consisting of three spatially separated components that takes as inputs bits $R_i$ from an $\epsilon$-SV source (for arbitrary initial $\epsilon < 0.5$) and outputs a bit $R_r$ such that under the assumption that the device behavior is as predicted by quantum theory, and that one component is fully trusted and characterised, the bit $R_r$ is certified to be arbitrarily free, except with negligible probability.
\end{theorem}

We defer a full security proof to future work \cite{RBH21} for the following reason. It may be expected that the Protocol II can be suitably modified to allow for a constant rate of noise. However, such a modification would require the application of a (strong) randomness extractor at the end of the protocol rather than obtaining the output bit from the outcomes of one randomly chosen run. It is then required to study if the techniques of the proofs against quantum side information from \cite{ADF+18, KAF17} can be adapted to prove general security in the 1sDI scenario. In the Supplemental Material, we show the existence of extremal sequential non-signalling assemblages that admit quantum realization in the time-ordered no-signaling (TONS) scenario. Furthermore, we show that such extremal assemblages are such that the outcomes of each party are perfectly random with respect to any non-signalling adversary. As such, under the correctness assumption that the devices behave according to quantum theory, picking the output bit from the measurement outcomes in a randomly chosen run (using an $\epsilon$-SV source) guarantees that the output bit is uniformly random and private with respect to any non-signalling adversary. Furthermore, since quantum correlations achieve algebraic violation of the steering inequality used in this protocol, the certification of perfect randomness can be made for arbitrary measurement independence, i.e., for arbitrary initial $\epsilon < 0.5$.

\textit{Discussion.-}
While the 1sDI protocol here considers a steering scenario with three honest parties (with a single characterised qubit system held by one of the parties). The reasons for the choice of the three-party setting are as follows. Firstly, while it is possible to realise extremal no-signaling assemblages already in the two-party scenario (by virtue of the fact that there are no super-quantum no-signaling assemblages in this scenario), none of the extremal points so far have been found to give rise to a perfect random bit under the constraint of arbitrarily weak measurement independence. On the other hand, such a random bit can be readily obtained from a three-party test on a GHZ state for arbitrary initial $\epsilon$. Secondly, as remarked earlier, it was shown in \cite{RBRH20, BMRH21} that quantum correlations allow to realise non-trivial extremal no-signaling assemblages in the three-party scenario. Thirdly, one may readily anticipate that the 1sDI protocol for randomness amplification may be generalised into one for quantum key distribution \cite{BCW+12} or for conference key agreement \cite{RMW18}.   

\textit{Open Questions.-} 
In the future, it would be important to prove or disprove if fully device-independent randomness amplification of public weak sources against no-signaling adversaries is possible in a protocol using finite devices and settings, for arbitrary initial randomness of the source. In this regard, it is important to verify whether improved statistical techniques can improve the range of initial $\epsilon$ that can be amplified using Protocol I. It would also be interesting to show device-independent security against so-called hybrid quantum-no-signalling adversaries \cite{BMRH21} who are restricted to preparing devices whose input-output statistics conform to the laws of quantum mechanics on a subset of the parties. Remarkably, this new paradigm differs from the one-sided-device-independent scenario in that only the measurement statistics are required to behave quantumly and no trust is placed on any of the subsystems. Finally, it is important to use our protocols as a building block to design device- and one-sided device-independent randomness expansion and key distribution protocols from weak public sources of randomness.

{\it Acknowledgments.-}
R. R. acknowledges support from the Start-up Fund "Device-Independent Quantum Communication Networks" from The University of Hong Kong, the Seed Fund "Security of Relativistic Quantum Cryptography" (Grant No. 201909185030) and the Early Career Scheme (ECS) grant "Device-Independent Random Number Generation and Quantum Key Distribution with Weak Random Seeds" (Grant No. 27210620). P.H. and M.B. acknowledge support from the Foundation for Polish Science (IRAP project, ICTQT, contract no. 2018/MAB/5, co-financed by EU within Smart Growth Operational Programme).




\section{Supplemental Material: Proof of Security of the Device-Independent Randomness Amplification Protocol I}
Here, we give the formal proof of composable security for the device-independent protocol I for randomness amplification from public SV sources presented in the main text. 

\subsection{The Bell inequality.}
\label{sec:Bell-ineq}
The Bell inequality tested in a device-independent randomness amplification protocol against no-signaling adversaries is required to obey specific properties: (i) it should not be possible to simulate the quantum value of the Bell expression using classical strategies, under any bias $\epsilon$ of the source; 
(ii) the Bell test must certify randomness against a no-signaling adversary, i.e., there must exist a hashing function $f: |\textbf{A}| \times |\textbf{B}| \times |\textbf{X}| \times |\textbf{Y}| \rightarrow \{0,1\}$ which when applied to any no-signaling box $P_{\textbf{A}, \textbf{B} | \textbf{X}, \textbf{Y}}(\textbf{a},\textbf{b}| \textbf{x}, \textbf{y})$ that achieves the quantum value, obeys $P(f(\textbf{a}, \textbf{b}, \textbf{x}, \textbf{y}) | \textbf{x}, \textbf{y}) < 1$

In Protocol I, we test for non-local quantum correlations by means of a modified version of Hardy's ladder test of non-locality \cite{Hardy93, Hardy97, BBMH97, Cereceda02, Cereceda15}. This is a tilted version of the chained Bell inequality used in \cite{CR12} that allows for arbitrary measurement independence \cite{PRB+14}, besides satisfying strong monogamy properties \cite{RH14}. In the test, we consider two honest parties Alice and Bob who each perform measurements of $N+1$ dichotomic observables. That is, in each run of the protocol, Alice and Bob choose measurements given by $\textbf{x}, \textbf{y} \in \{0, \dots, N\}$ respectively and obtain the corresponding outcomes $\textbf{a}, \textbf{b} \in \{0,1\}$. The following Hardy conditions then yield a contradiction between quantum mechanics and local realistic theories.
 \begin{eqnarray}
 \label{eq:Hardy-zeros}
&&P_{\textbf{A}, \textbf{B}| \textbf{X}, \textbf{Y}}(0,1|k,k-1) = 0 \; \; \; \; \forall k \in \{1,\dots, N\} \nonumber \\
&&P_{\textbf{A}, \textbf{B} | \textbf{X}, \textbf{Y}}(1,0|k-1,k) = 0 \; \; \; \; \forall k \in \{1, \dots, N\}  \nonumber \\
&&P_{\textbf{A}, \textbf{B}| \textbf{X}, \textbf{Y}}(0,0|0,0) = 0,
\end{eqnarray}
 and 
 \begin{equation}
 \label{eq:Hardy-prob}
P_{\text{H}} = P_{\textbf{A}, \textbf{B}| \textbf{X}, \textbf{Y}}(0,0|N,N) \neq  0.
 \end{equation}
It can be readily verified that for all local deterministic boxes, the conditions in Eq.\eqref{eq:Hardy-zeros} guarantee that the probability $P_{\textbf{A}, \textbf{B}| \textbf{X}, \textbf{Y}}(0,0|N,N) = 0$, so that a non-zero value of this quantity indicates a contradiction with the notion of local realism. 

The amount of violation of local realism can be quantified in terms of the expression
\begin{widetext}
\begin{eqnarray}
\label{eq:Hardy-ineq-1}
\mathcal{I}\left(P_{\textbf{A}, \textbf{B}| \textbf{X}, \textbf{Y}} \right) &=& \left[\frac{1}{2} - P_{\textbf{A}, \textbf{B}| \textbf{X}, \textbf{Y}}(0,0|N,N) \right] + \sum_{k = 1}^{N} \bigg[ P_{\textbf{A}, \textbf{B}| \textbf{X}, \textbf{Y}}(0,1|k,k-1) + P_{\textbf{A}, \textbf{B}| \textbf{X}, \textbf{Y}}(1,0|k-1,k) \bigg]  \nonumber \\ && +P_{\textbf{A}, \textbf{B}| \textbf{X}, \textbf{Y}}(0,0|0,0) \nonumber \\
\end{eqnarray}
\end{widetext}
The minimum value of $\mathcal{I}\left(P_{\textbf{A}, \textbf{B}| \textbf{X}, \textbf{Y}} \right)$ over all local realistic boxes $\big\{P_{\textbf{A}, \textbf{B}| \textbf{X}, \textbf{Y}}\big\}$ is $\frac{1}{2}$, while the minimum value over all boxes that obey the no-signaling constraints is $0$. The optimal no-signaling strategy is achieved by a generalization of the Popescu-Rohrlich box \cite{Rohlich-Popescu} which assigns $P_{\textbf{A}, \textbf{B}| \textbf{X}, \textbf{Y}}(0,0|N,N) = \frac{1}{2}$ and value $0$ to the rest of the probabilities that appear in Eq.\eqref{eq:Hardy-ineq-1}. 

Another way to view this inequality is as a set of constraints that impose that all probabilities in the Hardy set be zero. Specifically, these Hardy constraints are given by setting to zero the quantity
\begin{widetext}
\begin{eqnarray}
\label{eq:Hardy-ineq-2}
\mathcal{I}_0\left(P_{\textbf{A}, \textbf{B}| \textbf{X}, \textbf{Y}} \right) = \sum_{k = 1}^{N} \bigg[ P_{\textbf{A}, \textbf{B}| \textbf{X}, \textbf{Y}}(0,1|k,k-1) + P_{\textbf{A}, \textbf{B}| \textbf{X}, \textbf{Y}}(1,0|k-1,k) \bigg] + P_{\textbf{A}, \textbf{B}| \textbf{X}, \textbf{Y}}(0,0|0,0),
\end{eqnarray}
\end{widetext}
which implies in a local realistic theory the constraint $P_{\textbf{A}, \textbf{B}| \textbf{X}, \textbf{Y}}(0,0|N,N) = 0$, while this last probability can reach close to the value $1/2$ in quantum theory under the same Hardy constraints. 

The optimal quantum strategy achieves a value close to $1/2$ for the probability $P_{\textbf{A}, \textbf{B}| \textbf{X}, \textbf{Y}}(0,0|N,N)$, and thus a value close to $0$ for the quantity $\mathcal{I}\left(P_{\textbf{A}, \textbf{B}| \textbf{X}, \textbf{Y}} \right)$. Specifically, Alice and Bob share copies of the state 
\begin{eqnarray}
| \psi \rangle = \alpha |00 \rangle - \beta | 11 \rangle,
\end{eqnarray}
where $\alpha, \beta \in \mathbb{C}$ are parameters depending on $N$, that obey $|\alpha|^2 + | \beta|^2 = 1$. Alice and Bob perform on their half of the shared state,  measurements given by the following observables
\begin{eqnarray}
A_k &=& \sum_{j = 0, 1} (-1)^j \Pi_{k}^{j}, \nonumber \\
B_k &=& \sum_{j=0, 1} (-1)^j \Sigma_{k}^{j}.
\end{eqnarray}
Here, the projectors $\Pi_k^j = |\pi_k^j \rangle \langle \pi_k^j|$ and $\Sigma_k^j = |\sigma_k^j \rangle \langle \sigma_k^j|$ are in general given by
\begin{eqnarray}
| \pi_k^0 \rangle &=& \cos a_k | 0 \rangle + \sin a_k | 1 \rangle, \nonumber \\
| \pi_k^1 \rangle &=& - \sin a_k | 0 \rangle + \cos a_k | 1 \rangle,
\end{eqnarray}
and 
\begin{eqnarray}
| \sigma_k^0 \rangle &=& \cos b_k | 0 \rangle + \sin b_k | 1 \rangle, \nonumber \\
| \sigma_k^1 \rangle &=& - \sin b_k | 0 \rangle + \cos b_k | 1 \rangle,
\end{eqnarray}
for angles $a_k, b_k \in [0, 2 \pi]$. In order to set all the Hardy constraint probabilities to zero, we obtain the conditions 
\begin{eqnarray}
\tan a_k \cot b_{k-1} &=& - \frac{\alpha}{\beta}, \; \; \forall k = 1, \dots, N \nonumber \\
\cot a_{k-1} \tan b_k &=& - \frac{\alpha}{\beta}, \; \; \forall k = 1, \dots, N \nonumber \\
\tan a_0 \tan b_0 &=& \frac{\alpha}{\beta},
\end{eqnarray}
for any fixed parameters $\alpha, \beta$ defining the state. These conditions give that
\begin{eqnarray}
\tan b_N \tan a_N &=& \left( \frac{\alpha}{\beta} \right)^{2N+1}.
\end{eqnarray}
With these parameters, we obtain the Hardy probability $P_{\text{H}} = P_{\textbf{A}, \textbf{B}| \textbf{X}, \textbf{Y}}(0,0|N,N)$ to be given by
\begin{eqnarray}
P_{\text{H}} = P_{\textbf{A}, \textbf{B}| \textbf{X}, \textbf{Y}}(0,0|N,N) = \alpha^2 \left[1 - \left(\frac{\alpha}{\beta} \right)^{2N} \right]^2 \frac{\cos^2 a_N}{1 + \left(\frac{\alpha}{\beta}\right)^{4N+2} \cot^2 a_N }.
\end{eqnarray}
For fixed $\alpha, \beta$, we find the optimal value of $a_N$ to maximize $P_{\text{H}}$. Differentiating and setting $P_{\text{H}}$ to zero gives the optimal value to be
\begin{eqnarray}
a_N = \arctan \left[ \left( \frac{\alpha}{\beta} \right)^{N+1/2} \right].
\end{eqnarray}
This then gives that $\tan a_k = \tan b_k$ for all $k = 0, \dots, N$ or in other words that $b_k = a_k + m \pi$ for each $k$. Substituting gives the Hardy probability to be
\begin{eqnarray}
P_{\text{H}} = \max_{\alpha, \beta} \left( \frac{\alpha \beta^{2N+1} - \beta \alpha^{2N+1}}{\beta^{2N+1} + \alpha^{2N+1}} \right)^2,
\end{eqnarray}
which is similar to the expression originally derived by Hardy for this ladder scenario. Hardy had argued from this expression that $P_{\text{H}}$ tends to $1/2$ in the limit of large $N$, with the state getting close to a maximally entangled state. For our purposes though, this is not sufficient, we need to identify the exact dependence on $N$ of the expression $1/2 - P_{\text{H}}$. 

To this end, let us set $x = \alpha/\beta$ and write 
\begin{eqnarray}
P_{\text{H}} = \max_{0 \leq x \leq 1} \frac{x^2}{1+x^2} \left(\frac{1 - x^{2N}}{1+x^{2N+1}} \right)^2.
\end{eqnarray}
This gives that
\begin{eqnarray}
\mathcal{I}\left(P_{\textbf{A}, \textbf{B}| \textbf{X}, \textbf{Y}} \right) = \min_{0 \leq x \leq 1} \frac{1}{2} -  \frac{x^2}{1+x^2} \left(\frac{1 - x^{2N}}{1+x^{2N+1}} \right)^2.
\end{eqnarray}
To identify the behavior of the optimal value of $\mathcal{I}\left(P_{\textbf{A}, \textbf{B}| \textbf{X}, \textbf{Y}} \right)$ as a function of $N$, we will use an ansatz $x^*(N) = c^{N^{-d}}$ for reals $0 < c, d < 1$ so that
\begin{eqnarray}
\mathcal{I}\left(P_{\textbf{A}, \textbf{B}| \textbf{X}, \textbf{Y}} \right) &\leq & \frac{1}{2} -  \frac{(x^*(N))^2}{1+(x^*(N))^2} \left(\frac{1 - (x^*(N))^{2N}}{1+(x^*(N))^{2N+1}} \right)^2, \nonumber \\
& = &\frac{1}{2} -  \frac{c^{2N^{-d}}}{1+c^{2N^{-d}}} \left(\frac{1 - c^{2N^{1-d}}}{1+c^{2N^{1-d} + N^{-d}}} \right)^2.
\end{eqnarray}
Now, noting that $0 < c, d <1$ so that $c^{2N^{1-d}} \rightarrow 0$ and $c^{N^{-d}} \rightarrow 1$, we write
\begin{eqnarray}
\mathcal{I}\left(P_{\textbf{A}, \textbf{B}| \textbf{X}, \textbf{Y}} \right) &\leq &  \frac{1}{2} - \frac{c^{2N^{-d}}}{1+c^{2N^{-d}}} \left( 1 - c^{2N^{1-d}} - c^{2N^{1-d} + N^{-d}} + c^{4N^{1-d} + N^{-d}} \right)^2.
\end{eqnarray}
We thus get
\begin{eqnarray}
\mathcal{I}\left(P_{\textbf{A}, \textbf{B}| \textbf{X}, \textbf{Y}} \right) &\leq & \frac{-1}{2} \log_e c \; N^{-d} + \frac{1}{6} (\log_e c)^3 N^{-3d} - \frac{1}{15} (\log_e c)^5 N^{-5d} + O(N^{-7d}).
\end{eqnarray}
We can choose $d < 1$ to say $d = 0.99$. To fix the value of $c$, we use the value of $\mathcal{I}\left(P_{\textbf{A}, \textbf{B}| \textbf{X}, \textbf{Y}} \right)$ for $N = 1$. We know that in this case the Hardy probability is $p_H = \frac{5 \sqrt{5}-11}{2} \approx 0.09$ \cite{Hardy93} so that $\mathcal{I}\left(P_{A,B|X,Y} \right) = \frac{12 - 5 \sqrt{5}}{2} \approx 0.41$. This gives that $c \approx 0.4643$ so that we obtain 
\begin{eqnarray}
\mathcal{I}\left(P_{\textbf{A}, \textbf{B}| \textbf{X}, \textbf{Y}} \right) = 0.384 N^{-d} + O(N^{-3d}),
\end{eqnarray}
with $d = 0.99$, where we used $(-1/2) \log_e c \approx 0.384$. We thus see that the quantum value of the quantity $\mathcal{I}\left(P_{\textbf{A}, \textbf{B}| \textbf{X}, \textbf{Y}} \right)$ reaches its optimal no-signaling value of $0$ at the rate $N^{-d}$ with $d < 1$. 

\subsubsection{Measurement-dependent locality inequalities.}
\label{subsec:MDL-ineq}

We here present an intuitive explanation for why the ladder Hardy inequality considered in this article allows for the extraction of randomness from public SV sources, while previous approaches did not. The reason is that the ladder Hardy inequality allows to test for the violation of local realism with an arbitrarily small amount of measurement independence (aka free will). This is particularly evident when the ladder Hardy inequality is cast as a measurement-dependent locality (MDL) inequality \cite{PRB+14}.

The MDL inequalities were designed to certify quantum non-locality in the situation of limited measurement-independence. In our scenario where each party chooses from among $N+1$ binary measurements using an $\epsilon$-SV source, the MDL inequality corresponding to the ladder Hardy paradox is written as 
\begin{eqnarray}
\label{eq:MDL-gen}
&&\left( \frac{1}{2} - \epsilon \right)^{2 \log_2 (N+1)} P_{\textbf{A}, \textbf{B}, \textbf{X}, \textbf{Y}}(0,0,N,N) - \nonumber \\
&&\left( \frac{1}{2} + \epsilon \right)^{2 \log_2 (N+1)} \left( \sum_{k = 1}^{N} \bigg[ P_{\textbf{A}, \textbf{B}, \textbf{X}, \textbf{Y}}(0,1,k,k-1) + P_{\textbf{A}, \textbf{B}, \textbf{X}, \textbf{Y}}(1,0,k-1,k) \bigg] + P_{\textbf{A}, \textbf{B}, \textbf{X}, \textbf{Y}}(0,0,0,0) \right) \leq 0, \nonumber \\
\end{eqnarray}
Firstly note that this is an inequality over the joint distributions $P_{\textbf{A}, \textbf{B}, \textbf{X}, \textbf{Y}}$ rather than the conditional distribution $P_{\textbf{A}, \textbf{B}| \textbf{X}, \textbf{Y}}$ that we have been dealing with thus far. Secondly, note that by the Hardy constraints any classical box for which $P_{\textbf{A}, \textbf{B}| \textbf{X}, \textbf{Y}}(0,0|N,N) > 0$ satisfies the condition that at least one of the probabilities belonging to the Hardy constraint set is also non-zero. This implies that even when the $\epsilon$-SV source generates the input $(N,N)$ with probability $P_{\textbf{X}, \textbf{Y}}(N, N) = \left(\frac{1}{2} + \epsilon \right)^{2 \log (N+1)}$, the inequality in \eqref{eq:MDL-gen} cannot be violated by a classical box. On the other hand, the corresponding quantum strategies satisfy the Hardy constraints and achieve $P_{\textbf{A}, \textbf{B}| \textbf{X}, \textbf{Y}}(0,0|N,N) > 0$, thereby violating the inequality. In other words, the ladder Hardy inequality when cast in this MDL form, shows a violation of local realism for any value of $\epsilon$ in the range $[0, 1/2)$.  

This also means that if the honest parties observe a violation of the MDL inequality, then they can infer that the measurement outputs possess some min-entropy for any arbitrary value of $\epsilon \in [0, 1/2)$. At first glance, one might then suggest a randomness amplification protocol such as proposed in \cite{our-4} for our task of extracting randomness from public sources against no-signaling adversaries. However, it turns out that this is not the case. This is because, while the violation of an MDL inequality allows us to infer the presence of randomness (specifically some amount of min-entropy), we still require an extraction procedure to generate fully uniform private random bits from such a min-entropy source. As mentioned earlier, the quantification and extraction from min-entropy sources against a no-signaling adversary is a difficult problem, that requires the development of new tools. We therefore adopt a different procedure in this paper, in which the extraction procedure is explicitly given by the honest parties picking a run and generating the output random bit from the measurement outcomes in that run. Since one cannot infer the perfect randomness in measurement outcomes directly from the amount of MDL inequality violation, the procedures of \cite{our-4} do not directly apply.

\subsection{Composable Security Definition for Public Source Randomness Amplification} 

We begin by describing the scenario for the device-independent randomness amplification protocol. Let us recall that the $\epsilon$-SV source is defined by the condition that bits $b_i$ produced by the source obey
\begin{equation}
\label{eq:SVdef-sup}
\frac{1}{2} - \epsilon \leq \text{Pr}(b_i = 0 | b_{i-1}, \dots, b_1) \leq \frac{1}{2} + \epsilon
\end{equation}
for some $0 \leq \epsilon < \frac{1}{2}$. The honest parties and adversary Eve share a no-signaling box $\big\{P_{A,B,W|X',Y',Z}\big\}$ where $X' = \left(\textbf{X'}_1, \dots, \textbf{X'}_M \right)$, $Y' =\left(\textbf{Y'}_1, \dots, \textbf{Y'}_M \right)$ and $A = \left(\textbf{A}_1, \dots, \textbf{A}_M \right)$, $B = \left(\textbf{B}_1, \dots, \textbf{B}_M \right)$ are the random variables denoting the inputs and outputs respectively of the honest parties for the $M$ runs of the Protocol I. $W$ and $Z$ are the random variables denoting the output and input respectively of the adversary Eve. 

The honest parties draw bits $X, Y$ (of size $\log (N+1)$ for each of the $M$ runs) from the $\epsilon$-SV source to input into the box, i.e., they set $X' = X$ and $Y' = Y$. As defined, $\mathcal{S}_{\text{H}}$ is the subset of measurement runs in which the measurement settings involved in the Hardy ladder test \eqref{eq:Hardy-ineq-1} appear, i.e., $\mathcal{S}_{\text{H}} := \big\{ i \in \{1, \dots, M \}| (\textbf{x}_i, \textbf{y}_i) = (k, k-1) \; \text{for} \; k \in \{1, \dots, N\} \vee (\textbf{x}_i, \textbf{y}_i) = (k-1, k) \; \text{for} \; k \in \{1, \dots, N\} \vee (\textbf{x}_i, \textbf{y}_i) = (0, 0) \vee (\textbf{x}_i, \textbf{y}_i) = (N, N) \big\}$. The parties check that the size of the set $\mathcal{S}_{\text{H}}$ satisfies $|\mathcal{S}_{\text{H}}| \in \left[\frac{M}{(N+1)}, \frac{3M}{(N+1)}\right]$, and only accept if this condition is satisfied. 
We therefore obtain that 
\begin{eqnarray}
\text{P}_{\textbf{X}_i, \textbf{Y}_i}(\textbf{x}_i, \textbf{y}_i) \geq p_{\min}(\epsilon, N),
\end{eqnarray}
for all inputs $(\textbf{x}_i, \textbf{y}_i) \in \mathcal{S}_{\text{H}}$. 
Here 
\begin{eqnarray}
p_{\min}(\epsilon, N) := \frac{1}{2(N+1)} \left(\frac{1 - 2 \epsilon}{1 + 2 \epsilon} \right)^{2 \log (N+1)}
\end{eqnarray}
is the minimum probability of choosing any particular pair of inputs $\left(\textbf{X}_i, \textbf{Y}_i \right)$ from the $\epsilon$-SV source (normalized by the total probability of choosing the $2(N+1)$ inputs in the Hardy ladder test). 

The adversary has classical information $E$ correlated to $X, Y$ (note that in future, it would be interesting to consider no-signaling side information about the source as well). The box we consider for the protocol is given by the family of probability distributions $\big\{P_{A, B, W, X, Y, E | X', Y', Z} \big\}$.

The box $\big\{P_{A, B, W, X, Y, E | X', Y', Z} \big\}$ is only restricted by the constraints of no-signaling between Alice and Bob, and time-ordered no-signaling within Alice's (and Bob's) laboratory. Specifically, the assumptions are 
\begin{eqnarray}
\label{eq:ns1}
P_{A, B|X', Y', Z} &=& P_{A, B | X', Y'} \nonumber \\
P_{A, W| X', Y', Z} &=& P_{A, W | X', Z} \nonumber \\
P_{B, W | X', Y', Z} &=& P_{B, W | Y', Z}
\label{eq:as-nosig}
\end{eqnarray}

Each device component also obeys a time-ordered no-signaling (\texttt{TONS}) condition for the $k \in [M]$ runs performed on it:
\begin{eqnarray}
\label{eq:tons1}
&&P_{\textbf{A}_k | Z, W, X', Y', X, Y, E} = P_{\textbf{A}_k | Z, W, X'_{\leq k}, X, Y, E} \; \; \; \; \forall k \in [M], \nonumber \\
&&P_{\textbf{B}_k | Z, W, X', Y', X, Y, E} = P_{\textbf{B}_k | Z, W, Y'_{\leq k}, X, Y, E} \; \; \; \; \forall k \in [M].
\end{eqnarray}
Here, note that $\textbf{A}_k, \textbf{B}_k$ denote the outputs for the $k$-th run. 

In words, the main assumptions are that the different components of the device do not signal to each other and to the adversary Eve. Additionally, there is also a time-ordered no-signaling (\texttt{TONS}) structure assumed on different runs of a single party, the outputs in any run may depend on the previous inputs within the party's component but not on future inputs.

We make three additional assumptions that are standard in the literature on randomness amplification (except in the protocol of \cite{CSW14} which requires a polynomially large number of devices). 
\begin{enumerate}
\item \textbf{Assumption A1:} The probability distribution $P_{X,Y,E}$ satisfies the $\epsilon$-SV condition in \eqref{eq:SVdef-sup}.

\item \textbf{Assumption A2:} The devices do not signal to the SV source, i.e. the distribution of
$(X,Y, E)$ is independent of the inputs $(X',Y', Z)$, i.e.,
\begin{eqnarray}
\sum_{a, b, w} P_{A,B,W,X,Y,E|X',Y',Z}(a,b,w,x,y,e|x',y',z) = P_{X,Y,E}(x,y,e) \quad \forall (x,y,e,x',y',z)
\end{eqnarray}

\item \textbf{Assumption A3:} The box is fixed independently of the SV source, i.e.,
\begin{eqnarray}
P_{A,B,W|X',Y',W,X,Y,E}(a,b,w|x',y',z,x,y,e) &=& P_{A,B,W|X',Y',Z}(a,b,w|x',y',z) \nonumber \\ && \qquad \qquad \forall (a,b,w,x',y',z,x,y,e) 
\end{eqnarray}

\end{enumerate}

Assumption A1 comes from the fact that the random variables $X,Y,E$ correspond to the $\epsilon$-SV source. Assumption A2 is a natural assumption that clarifies the notion of inputting random variables ($X,Y$) from the $\epsilon$-SV source into the box (as $X', Y'$). Assumption A3 is the assumption of source-device independence, i.e., that the device does not change depending on bits extracted from the SV source. This assumption is also employed in most of the results on device-independent randomness amplification including \cite{CR12, BRGH+16, RBHH+16, GMTD+13} and can be viewed as the Markov assumption in \cite{KAF17}. The one result that does not consider this assumption is the protocol of \cite{CSW14}, in which as mentioned a polynomially large number of devices is needed. The assumption A3 may be viewed as the quantum analogue of the problem of extracting randomness from one source and an independent (but unknown and arbitrary) channel. While no randomness can be extracted if the channel is classical, the situation is different when considering quantum non-local correlations.


\subsubsection{Public SV source.}
Our assumptions as stated are the standard ones in the literature on the topic. In order to show that the final output bit $R_r$ is uniform and private when the bits $X,Y$ are chosen from the public $\epsilon$-SV source, we need to first consider the security definition, and in particular how the public source changes this definition. 
After inputting $X, Y$ as $X', Y'$, we obtain $\big\{ P_{A,B,W,X,Y,E|X,Y,Z} \big\}$. 
The first test in Protocol I checks that $\textit{Z}_{\text{H}}(A,B,X,Y) := 
\sum_{i \in \mathcal{S}_{\text{H}}} B_{\text{H}}(\textbf{a}_i, \textbf{b}_i, \textbf{x}_i, \textbf{y}_i) \leq \delta_1$, for fixed $\delta_1 > 0$.
Here $B_{\text{H}}(\textbf{a}_i, \textbf{b}_i, \textbf{x}_i, \textbf{y}_i)$ is an indicator function that takes value $1$ when $(\textbf{a}_i, \textbf{b}_i, \textbf{x}_i, \textbf{y}_i)$ corresponds to one of the 'zero' probabilities from the set of Hardy constraints, i.e., 
\begin{equation}
 \label{eq:Hardy-z}
B_{\text{H}}(\textbf{a}_i, \textbf{b}_i, \textbf{x}_i, \textbf{y}_i)=  \begin{cases} 

    1 & (\textbf{x}_i,\textbf{y}_i) = (k,k-1) \wedge (\textbf{a}_i, \textbf{b}_i) = (0,1) \; \; \; \; \forall k \in \{1, \dots, N\} \\
       1 & (\textbf{x}_i,\textbf{y}_i) = (k-1,k) \wedge (\textbf{a}_i, \textbf{b}_i) = (1,0) \; \; \; \; \forall k \in \{1, \dots, N\} \\
       1 & (\textbf{x}_i,\textbf{y}_i) = (0,0) \wedge (\textbf{a}_i, \textbf{b}_i) = (0,0) \\
      0 & \text{otherwise}
   \end{cases}
\end{equation}
We define the set of output-inputs $(A,B,X,Y)$ that pass this test for fixed $\delta_1 > 0$ as the set $\texttt{ACC}_1$, i.e., 
\begin{eqnarray}
\texttt{ACC}_1 :=  \bigg\{ (A,B,X,Y) \big| \textit{Z}_{\text{H}}(A,B,X,Y)  \leq \delta_1 \bigg\}.
\end{eqnarray}
The second test in Protocol I checks that 
\begin{eqnarray}
Z_{\text{H-nz}}(A,B,X,Y) := \frac{1}{|\mathcal{S}_{\text{H-nz}}|} \sum_{i \in \mathcal{S}_{\text{H-nz}}} B_{\text{H-nz}}(\textbf{a}_i, \textbf{b}_i, \textbf{x}_i, \textbf{y}_i) \geq \frac{1}{2} - \delta_2, \nonumber \\
\end{eqnarray}
for fixed $\delta_2 > 0$. Here, $\mathcal{S}_{\text{H-nz}}$ is the subset of measurement runs in which the measurement setting $(N,N)$ appears, i.e., $\mathcal{S}_{\text{H-nz}} := \big\{i \in \{1, \dots, M\} | (\textbf{x}_i, \textbf{y}_i) = (N,N) \big\}$. And $B_{\text{H-nz}}(\textbf{a}_i, \textbf{b}_i, \textbf{x}_i, \textbf{y}_i)$ is an indicator function that takes value $1$ when $(\textbf{a}_i, \textbf{b}_i, \textbf{x}_i, \textbf{y}_i)$ corresponds to the Hardy probability $P_{\text{H}}$, i.e., 
\begin{equation}
\label{eq:B-Hardy-nz}
B_{\text{H-nz}}(\textbf{a}_i, \textbf{b}_i, \textbf{x}_i, \textbf{y}_i)=  \begin{cases} 
      1 & (\textbf{x}_i, \textbf{y}_i) = (N,N) \wedge (\textbf{a}_i, \textbf{b}_i) = (0,0)  \\
      0 & \text{otherwise}
   \end{cases}
\end{equation}
We define the set of output-inputs $(A,B,X,Y)$ that pass this test for fixed $\delta_2 > 0$ as the set $\texttt{ACC}_2$, i.e., 
\begin{eqnarray}
\texttt{ACC}_2:=  \bigg\{ (A,B,X,Y) \big| \textit{Z}_{\text{H-nz}}(A,B,X,Y) \geq \frac{1}{2} - \delta_2 \bigg\}
\end{eqnarray}
We define the set $\texttt{ACC}$ as the set of output-inputs $(A,B,X,Y)$ that pass both tests of Protocol I, i.e., $\texttt{ACC} = \texttt{ACC}_1 \cap \texttt{ACC}_2$. Upon acceptance of the protocol, we obtain the family of probability distributions $P_{A,B,W,X,Y,E|Z, \texttt{ACC}}$. The final output bit $R_r$ is chosen as a function of the output bits $A,B$ after which we end up with the family of distributions $P_{R_r,W,X,Y,E|Z, \texttt{ACC}}$. In this case, the procedure to obtain $R_r$ is for the honest parties to choose a run $r \in S_{\text{H-nz}}$ using bits from the $\epsilon$-SV source and output the bit $R_r$ obtained from the measurement outcomes of the $r$-th run as
 \[R_r =  \begin{cases} 
      0 & (\textbf{a}_r, \textbf{b}_r) = (0,0) \\
      1 & (\textbf{a}_r, \textbf{b}_r) \neq (0,0)
   \end{cases}
\]

To quantify the quality of the output $R_r$, we use the universally composable distance $d_{\text{comp}}$ \cite{BHLMO05, RK05}. This is defined as
\begin{eqnarray}
d_{\text{comp}} := \sum_{s = 0,1} \sum_{x,y,e} \max_z \sum_{w} \bigg| P_{R_r, W,X,Y,E|Z, \texttt{ACC}}(s,w,x,y,e|z, \texttt{ACC}) - \frac{1}{2} P_{W,X,Y, E|Z, \texttt{ACC}}(w,x,y,e|z, \texttt{ACC}) \bigg| \nonumber \\
\end{eqnarray}
Here $s$ is the value of the final bit $R_r$ given to the distinguisher attempting to distinguish between the distribution $P_{R_r, W,X,Y,E|Z, \texttt{ACC}}(r,w,x,y,e|z, \texttt{ACC})$ from the ideal distribution $\frac{1}{2} P_{W, X,Y,E|Z, \texttt{ACC}}(w,x,y,e|z, \texttt{ACC})$. The maximization occurs over inputs $z$ belonging to the adversary and the fact that this maximization depends on the values of $x,y$ (as well as his side information $e$) reflects the fact that the $\epsilon$-SV source is public, i.e., the inputs $x,y$ become known to the adversary and the adversary can choose a different input $z$ depending on the $x, y$.

\subsection{Randomness of outcomes in individual runs of the device}

Statistical tests prescribed by the protocol are performed on the outcomes of measurements for given inputs. We shall show that the tests have the property that the final output bit $R_r$ from any no-signaling box $\big\{P_{A, B, W, X, Y, E | X', Y', Z} \big\}$ that passes the tests must necessarily be arbitrarily free.


From the SV source condition \eqref{eq:SVdef-sup}, the probability that the parties draw, in the $j$-th run any specific pair of measurement settings $(x_j,y_j)$ is bounded as
\begin{eqnarray}
\left(\frac{1}{2} - \epsilon \right)^{2 \log (N+1)} \leq P_{\textbf{X}_j, \textbf{Y}_j | X_{< j}, Y_{< j}}(\textbf{x}_j, \textbf{y}_j | x_{< j}, y_{< j}) \leq \left( \frac{1}{2} + \epsilon \right)^{2 \log (N+1)}, \nonumber \\
\end{eqnarray}
where $X_{< j} := \left(\textbf{X}_1, \dots, \textbf{X}_{j-1} \right)$ and $Y_{< j} := \left(\textbf{Y}_1, \dots, \textbf{Y}_{j-1} \right)$.
Since only $2(N+1)$ out of the $(N+1)^2$ settings appear in the inequality, the probability for any measurement setting within the Hardy set after the parties throw away the remaining measurements is bounded as 
\begin{eqnarray}
P_{\textbf{X}_j, \textbf{Y}_j | X_{< j}, Y_{< j}}(\textbf{x}_j, \textbf{y}_j | x_{< j}, y_{< j}) \geq p_{\min}(\epsilon, N).
\end{eqnarray}
 


We first consider how the quantities $\mathcal{I}\left(P_{\textbf{A}, \textbf{B}|\textbf{X}, \textbf{Y}} \right)$ and $\mathcal{I}_0\left(P_{\textbf{A}, \textbf{B}| \textbf{X}, \textbf{Y}} \right)$ can allow us to certify randomness in the device outcomes in individual runs of the device. 
\begin{lemma}
\label{lem:single-run-bound}
Consider a two-party no-signaling box $\{P_{\textbf{A}, \textbf{B}| \textbf{X}, \textbf{Y}} \}$ with $A, B \in \{0,1\}$ and $X, Y \in \{0, \dots, N\}$ which obeys the condition $\mathcal{I}_0\left(P_{\textbf{A}, \textbf{B} | \textbf{X}, \textbf{Y}} \right) \leq \kappa_1$ for some fixed $\kappa_1 > 0$, where $B_{\text{H}}(a_i, b_i, x_i, y_i)$ is given by Eq.(\ref{eq:Hardy-z}). Then we have that the Hardy probability $P_{\textbf{A}, \textbf{B}|\textbf{X}, \textbf{Y}}(0,0|N, N)$ satisfies
\begin{eqnarray}
\label{eq:single-bound}
P_{\textbf{A}, \textbf{B}|\textbf{X}, \textbf{Y}}(0,0|N, N) \leq \frac{1}{2} \left( 1 + \kappa_1 \right).
\end{eqnarray}
\end{lemma}
\begin{proof}
The proof goes through repeated applications of the two-party no-signaling constraints \eqref{eq:ns1}. We have
\begin{eqnarray}
2P_{\textbf{A}, \textbf{B}|\textbf{X}, \textbf{Y}}(0,0|N, N) & \leq & \sum_{b = 0,1} P_{\textbf{A}, \textbf{B}|\textbf{X}, \textbf{Y}}(0,b|N, N-1) + \sum_{a = 0,1} P_{\textbf{A}, \textbf{B}|\textbf{X}, \textbf{Y}}(a,0|N-1, N) \nonumber \\
& \leq & P_{\textbf{A}, \textbf{B}|\textbf{X}, \textbf{Y}}(0,1|N,N-1) + P_{\textbf{A}, \textbf{B}|\textbf{X}, \textbf{Y}}(1,0|N-1, N) \nonumber \\
&&\;  + \sum_{a=0,1} P_{\textbf{A}, \textbf{B}|\textbf{X}, \textbf{Y}}(a,0|N-2,N-1) + \sum_{b=0,1} P_{\textbf{A}, \textbf{B}|\textbf{X}, \textbf{Y}}(0, b|N-1, N-2), \nonumber \\
& \leq & P_{\textbf{A}, \textbf{B}|\textbf{X}, \textbf{Y}}(0,1|N,N-1) + P_{\textbf{A}, \textbf{B}|\textbf{X}, \textbf{Y}}(1,0|N-1, N) \nonumber \\
&& \; + P_{\textbf{A}, \textbf{B}|\textbf{X}, \textbf{Y}}(1,0|N-2, N-1) + P_{\textbf{A}, \textbf{B}|\textbf{X}, \textbf{Y}}(0,1|N-1, N-2) \nonumber \\
&&\; + \sum_{b=0,1} P_{\textbf{A}, \textbf{B}|\textbf{X}, \textbf{Y}}(0, b|N-2, N-3) +\sum_{a=0,1} P_{\textbf{A}, \textbf{B}|\textbf{X}, \textbf{Y}}(a,0|N-3, N-2), \nonumber \\
&\leq & \dots \nonumber \\
&\leq & \sum_{k=1}^{N} \left[ P_{\textbf{A}, \textbf{B}|\textbf{X}, \textbf{Y}}(0,1|k,k-1) + P_{\textbf{A}, \textbf{B}|\textbf{X}, \textbf{Y}}(1,0|k-1,k) \right] + P_{\textbf{A}, \textbf{B}|\textbf{X}, \textbf{Y}}(0,0|0,0) \nonumber \\
&&\; P_{\textbf{A}, \textbf{B}|\textbf{X}, \textbf{Y}}(0,0|0,0) + P_{\textbf{A}, \textbf{B}|\textbf{X}, \textbf{Y}}(0,1|0,0) + P_{\textbf{A}, \textbf{B}|\textbf{X}, \textbf{Y}}(1,0|0,0) \nonumber \\
& \leq & 1 + \mathcal{I}_0\left(P_{\textbf{A}, \textbf{B}|\textbf{X}, \textbf{Y}}\right) \nonumber \\
& \leq & 1 + \kappa_1,
\end{eqnarray}
giving 
\begin{eqnarray}
P_{\textbf{A}, \textbf{B}|\textbf{X}, \textbf{Y}}(0,0|N, N) \leq \frac{1}{2} \left( 1 + \kappa_1 \right).
\end{eqnarray}
Note that in each step of the proof, we have used two no-signaling constraints \eqref{eq:ns1} corresponding to the two parties' output distribution being independent of the other party's input. The ellipses denote repeated application of the no-signaling constraints as in the previous steps to show the expression for general $N$.  
\end{proof}

\subsection{Amplification under the Quantum assumption}
Let us first consider the ideal scenario that $P_{\textbf{A}, \textbf{B}|\textbf{X}, \textbf{Y}}$ are exactly as predicted by Quantum Theory. This assumption (denoted as assumption $QT$ in \cite{CR12}) was used to show an amplification protocol for $\epsilon$-SV sources with $\epsilon < (\sqrt{2}-1)^2/2 \approx 0.086$ in \cite{CR12} which was subsequently improved to $\epsilon < 0.0961$ in \cite{GHH+14}. We first show that under the assumption of no-signaling and quantum theory, the $\epsilon$-SV source can be amplified to produce perfectly uniform bits for arbitrary initial $\epsilon < 0.5$ using the ladder Hardy correlations shown in the previous subsections. That is, we show the following Proposition.

\begin{proposition}
\label{prop:quant-assum}
There exists a protocol using a device consisting of two spatially separated components that takes as input bits $R_i$ from an $\epsilon$-SV source (of arbitrary fixed $\epsilon < 0.5$) and outputs a bit $R_r$ such that under the assumptions of no-signaling and quantum theory, the bit $R_r$ is arbitrarily free except with negligible probability.
\end{proposition} 
\begin{proof}
Firstly, we observe that from Lemma \ref{lem:single-run-bound} we have that for any input-output pair $(w,z)$ of Eve, it holds that
\begin{eqnarray}
P_{\textbf{A}, \textbf{B} | \textbf{X}, \textbf{Y}, W, Z}(0,0|N, N, w, z) \leq \frac{1}{2} \left[1 + \mathcal{I}_0\left(P_{\textbf{A}, \textbf{B} | \textbf{X}, \textbf{Y}, w, z} \right) \right].
\end{eqnarray}
We can now derive
\begin{eqnarray}
\label{eq:marg-conditional}
\mathcal{I}_0\left(P_{\textbf{A}, \textbf{B}| \textbf{X}, \textbf{Y}} \right) &=& \sum_{w,z} \sum_{k = 1}^{N} \bigg[ P_{W, Z | k, k-1}(w,z) P_{\textbf{A}, \textbf{B}| \textbf{X}, \textbf{Y}, W, Z}(0,1|k,k-1, w, z)  \nonumber \\  && \qquad  +P_{W, Z | k-1, k}(w,z) P_{\textbf{A}, \textbf{B}| \textbf{X}, \textbf{Y}, W, Z}(1,0|k-1,k, w, z) \bigg]  + P_{W, Z | 0, 0}(w,z) P_{\textbf{A}, \textbf{B}| \textbf{X}, \textbf{Y}, W, Z}(0,0|0,0, w, z) \nonumber \\
&\geq & \min_{\textbf{x}, \textbf{y}, w, z} \frac{P_{W,Z|\textbf{X}, \textbf{Y}}(w,z)}{P_{W,Z|N, N}(w,z)}  \sum_{w,z} P_{W, Z|N, N}(w,z) \mathcal{I}_0\left(P_{\textbf{A}, \textbf{B} | \textbf{X}, \textbf{Y}, w,z} \right).
\end{eqnarray}
Following \cite{CR12}, we may denote the quantity measuring how free the measurement settings are as 
\begin{eqnarray}
\alpha(N, N) := \min_{\textbf{x}, \textbf{y}, w, z} \frac{P_{\textbf{X}, \textbf{Y} | w, z}(\textbf{x}, \textbf{y})}{P_{\textbf{X}, \textbf{Y}|w,z}(N, N)} \geq \left(\frac{1/2 -  \epsilon}{1/2+  \epsilon} \right)^{2 \log (N+1)},
\end{eqnarray}
so that 
\begin{eqnarray}
\frac{\mathcal{I}_0\left(P_{\textbf{A}, \textbf{B}| \textbf{X}, \textbf{Y}} \right)}{\alpha(N, N)} &\geq & \sum_{w,z} P_{W, Z|N, N}(w,z) \left[ 2 P_{\textbf{A}, \textbf{B}| \textbf{X}, \textbf{Y}, W, Z}(0,0|N, N, w, z) - 1 \right] \nonumber \\
& \geq & 2 P_{\textbf{A}, \textbf{B}| \textbf{X}, \textbf{Y}}(0,0|N, N) - 1,
\end{eqnarray}
or equivalently
\begin{equation}
P_{\textbf{A}, \textbf{B}| \textbf{X}, \textbf{Y}}(0,0|N, N) = \sum_{w,z} P_{R_r, W, Z|\textbf{X}, \textbf{Y}}(0, w, z|N,N) \leq \frac{1}{2} \left[1 + \frac{\mathcal{I}_0\left(P_{\textbf{A}, \textbf{B}| \textbf{X}, \textbf{Y}} \right)}{\alpha(N, N)}  \right],
\end{equation}
where recall that 
 \[R_r =  \begin{cases} 
      0 & (\textbf{a}, \textbf{b}) = (0,0) \wedge  (\textbf{x}, \textbf{y}) = (N, N)\\
      1 & (\textbf{a}, \textbf{b}) \neq (0,0) \wedge  (\textbf{x}, \textbf{y}) = (N, N)
   \end{cases}
\]
Note that crucially, in contrast to the correlations from the chain inequality considered in \cite{CR12, GHH+14}, here quantum correlations achieve exactly $\mathcal{I}_0\left(P_{\textbf{A}, \textbf{B}| \textbf{X}, \textbf{Y}} \right) = 0$, so that the upper bound of $1/2$ holds irrespective of the value of $\epsilon$. 

On the other hand, the second test works on the runs when the single input $(\textbf{x}, \textbf{y}) = (N,N)$ was chosen. We see that if
\begin{eqnarray}
P_{\textbf{A}, \textbf{B}| \textbf{X}, \textbf{Y}, W, Z}(0,0|N, N, w, z) \geq \frac{1}{2} - \delta,
\end{eqnarray}
for any fixed $\delta > 0$, then it also holds that
\begin{eqnarray}
P_{\textbf{A}, \textbf{B}| \textbf{X}, \textbf{Y}}(0,0|N, N) = \sum_{w,z} P_{R_r, W, Z|\textbf{X}, \textbf{Y}}(0, w, z|N,N) \geq \frac{1}{2} - \delta.
\end{eqnarray}
Noting that quantum correlations achieve $\delta = O(N^{-d})$ with $d = 0.99$, we have that under the assumptions of no-signaling and quantum theory, 
\begin{eqnarray}
\label{eq:dist-final-bit}
D\left(P_{R_r, W, Z|N,N}, P_{U} \times P_{W, Z| N, N}\right) &=& \sum_{s, w, z} \bigg|P_{R_r, W, Z|N, N}(s, w, z) - (1/2) \times P_{W, Z|N,N}(w,z) \bigg| \nonumber \\
& =& \sum_{w,z} P_{W,Z|N,N}(w,z) \bigg| P_{R_r|N, N, w, z}(s) - (1/2) \bigg| \nonumber \\
&\leq &  \sum_{w,z} P_{W,Z|N,N}(w,z) \max \bigg\{ \delta, \mathcal{I}_0\left(P_{\textbf{A}, \textbf{B}|\textbf{X}, \textbf{Y}, w, z} \right) \bigg\} \nonumber \\
& \leq & \max \bigg\{ \delta, \frac{ \mathcal{I}_0\left(P_{\textbf{A}, \textbf{B} | \textbf{X}, \textbf{Y}} \right)}{\alpha(N, N)} \bigg\} \nonumber \\
& \leq & \max \bigg\{ \delta, \mathcal{I}_0 \left(P_{\textbf{A}, \textbf{B} | \textbf{X}, \textbf{Y}} \right) \left(\frac{1/2 +  \epsilon}{1/2 - \epsilon} \right)^{2 \log (N+1)} \bigg\}
\end{eqnarray}
where $D$ denotes the variational distance $D\left(P_{A}, Q_A \right)  =  \sum_{a = 1}^{|A|} \big|P_A(a) - Q_A(a) \big|$ and $P_{U}$ denotes the uniform distribution, and the last inequality comes from Eq.\eqref{eq:marg-conditional}. 
Thus, in the ideal scenario when the correlations are exactly as predicted by quantum theory, the Protocol allows for the extraction of a perfectly uniform bit for any arbitrarily weak source (i.e., for any initial $\epsilon < 0.5$).
\end{proof}

\subsection{Estimation}
In this subsection, we work with the conditional distribution $P_{\textbf{A}, \textbf{B} | \textbf{X}, \textbf{Y}, W=w, Z=z}$ for arbitrary but fixed $w,z$ and show that the final bit for this distribution is close to uniform, similar to the approach followed in \cite{CR12}. Then we will use Prop. \ref{prop:quant-assum} to show that the same holds also for the marginal distribution $P_{\textbf{A}, \textbf{B} | \textbf{X}, \textbf{Y}}$. For ease of notation, we do not explicitly include the conditionals $W = w, Z =z$ in the rest of this subsection and take them to be implicitly given. For instance, we will estimate $\mathcal{I}_0\left(P_{\textbf{A}_i, \textbf{B}_i | \textbf{X}_i, \textbf{Y}_i, W = w, Z =z} \right)$ dropping $W =w, Z = z$.

The first test in Protocol I checks that $\textit{Z}_{\text{H}}(A,B,X,Y) := \frac{1}{| \mathcal{S}_{\text{H}}|} \sum_{i \in \mathcal{S}_{\text{H}}} B_{\text{H}}(\textbf{a}_i, \textbf{b}_i, \textbf{x}_i, \textbf{y}_i) \leq \delta_1$,
for fixed $\delta_1 > 0$, 
where $B_{\text{H}}(\textbf{a}_i, \textbf{b}_i, \textbf{x}_i, \textbf{y}_i)$ is defined in Eq.(\ref{eq:Hardy-z}). In words, $B_{\text{H}}(\textbf{a}_i, \textbf{b}_i, \textbf{x}_i, \textbf{y}_i)$ is an indicator that takes value $1$ whenever the measurement setting and outcome correspond to one of the Hardy 'zero' constraints, and takes value $0$ otherwise. 
Note that quantum correlations detailed in Sec. \ref{sec:Bell-ineq} achieve $\textit{Z}_{\text{H}}(A,B,X,Y) = 0$ for the Hardy zero constraints. 

Let us define the set of 'bad' runs $\mathcal{S}_{\text{H-bad}}$ for this test as 
\begin{eqnarray}
\mathcal{S}_{\text{H-bad}} := \big\{ i \in \mathcal{S}_{\text{H}} \big|  \mathcal{I}_0\left(P_{\textbf{A}_i, \textbf{B}_i | \textbf{X}_i, \textbf{Y}_i} \right) > \kappa \big\},
\end{eqnarray}
for fixed $\kappa > 0$. Recall that $ \mathcal{S}_{\text{H}}$ denotes the set of runs in which one of the $2N$ Ladder Hardy settings appear and is bounded as $M/(N+1) \leq |\mathcal{S}_{\text{H}}| \leq 3M/(N+1)$. We will choose $\kappa = O(N^{-t})$ for a constant $t > 0$, specifically we set $t = 0.99$. 

We note the following Freedman inequality, which is a martingale extension of the Bernstein inequality. In particular, the inequality shows that the large-deviation behaviour of a martingale is controlled by the predictable quadratic variation
and a uniform upper bound for the martingale difference sequence. 

\begin{theorem}[\cite{Fre75}]
\label{thm:Freedman-ineq}
Consider a real-valued martingale $\{X_k : k = 0, 1, 2,  \dots \}$ with difference sequence $\{X_{k} - X_{k-1} : k = 1, 2, 3, \dots\}$, and assume that the difference sequence is uniformly bounded, i.e., $|X_{k} - X_{k-1}| \leq R$ almost surely for $k = 1, 2, 3, \dots$. Define the predictable quadratic variation process of the martingale: 
\begin{eqnarray}
W_k := \sum_{j=1}^{k} \mathbb{E}_{j-1} \left(X_{j} - X_{j-1} \right)^2, \quad \text{for} \; k = 1, 2, 3, \dots.
\end{eqnarray}
Then, for all $\beta \geq 0$ and $\sigma^2 > 0$, it holds that
\begin{eqnarray}
\text{Pr}\large\{ \exists k \geq 0: X_k \geq \beta \; \wedge \; W_k \leq \sigma^2 \large\} \leq \exp \left\{ \frac{-\beta^2/2}{\sigma^2 + R \beta/3} \right\}.
\end{eqnarray}
\end{theorem}

We now make use of Theorem \ref{thm:Freedman-ineq} in estimating the number of bad runs $|\mathcal{S}_{\text{H-bad}}|$. Firstly, let us define random variables $U_i = \left(\textbf{a}_i, \textbf{b}_i, \textbf{x}_i, \textbf{y}_i \right)$ for $i \in \mathcal{S}_{\text{H}}$ and $U_0 = (z,e)$. Let us denote for $i \in \mathcal{S}_{\text{H}}$ the conditional expectation values 
\begin{eqnarray}
\overline{B}^{SV}_{\text{H}}(\textbf{a}_i, \textbf{b}_i, \textbf{x}_i, \textbf{y}_i) &=& \mathbb{E}\big(B_{\text{H}}(\textbf{a}_i, \textbf{b}_i, \textbf{x}_i, \textbf{y}_i) | U_{i-1}, \dots, U_1, U_0 \big) \nonumber \\
&=& \sum_{l=1}^{N} \bigg[ \nu^{SV}_{\textbf{X}_i, \textbf{Y}_i | \textbf{x}_{< i}, \textbf{y}_{<i}, \textbf{a}_{<i}, \textbf{b}_{<i} }(l,l-1) P_{\textbf{A}_i, \textbf{B}_i| \textbf{X}_i, \textbf{Y}_i, \textbf{x}_{< i}, \textbf{y}_{<i}, \textbf{a}_{<i}, \textbf{b}_{<i} }(0,1|l,l-1)  \nonumber \\
&& \qquad +\nu^{SV}_{\textbf{X}_i, \textbf{Y}_i | \textbf{x}_{< i}, \textbf{y}_{<i}, \textbf{a}_{<i}, \textbf{b}_{<i} }(l-1,l) P_{\textbf{A}_i, \textbf{B}_i| \textbf{X}_i, \textbf{Y}_i, \textbf{x}_{< i}, \textbf{y}_{<i}, \textbf{a}_{<i}, \textbf{b}_{<i} }(1,0|l,l-1) \bigg]  \nonumber \\ && \qquad + \nu^{SV}_{\textbf{X}_i, \textbf{Y}_i | \textbf{x}_{< i}, \textbf{y}_{<i}, \textbf{a}_{<i}, \textbf{b}_{<i} }(0,0) P_{\textbf{A}_i, \textbf{B}_i| \textbf{X}_i, \textbf{Y}_i, \textbf{x}_{< i}, \textbf{y}_{<i}, \textbf{a}_{<i}, \textbf{b}_{<i} }(0,0|0,0),
\end{eqnarray}
with $\nu^{SV}_{\textbf{X}_i, \textbf{Y}_i | \textbf{x}_{< i}, \textbf{y}_{<i}, \textbf{a}_{<i}, \textbf{b}_{<i} }(\textbf{x}_i, \textbf{y}_i)$ being the probability of choosing the measurement settings $(\textbf{x}_i, \textbf{y}_i)$ using the $\epsilon$-SV source and $\textbf{x}_{< i}, \textbf{y}_{<i}, \textbf{a}_{<i}, \textbf{b}_{<i}$ denoting the inputs and outputs for the runs $1, \dots, i-1$.
Within the set of runs $\mathcal{S}_{\text{H}}$, define for $k = 1, 2, 3, \dots$, the empirical average 
\begin{eqnarray}
L_k := \frac{1}{k} \sum_{i = 1}^{k} B_{\text{H}}(\textbf{a}_i, \textbf{b}_i, \textbf{x}_i, \textbf{y}_i),
\end{eqnarray}
and the arithmetic average of conditional expectation values
\begin{eqnarray}
\overline{L}_k =  \frac{1}{k} \sum_{i = 1}^{k} \overline{B}^{SV}_{\text{H}}(\textbf{a}_i, \textbf{b}_i, \textbf{x}_i, \textbf{y}_i).
\end{eqnarray}
Now, define $X_0 := 0$ and $X_k = k(L_k - \overline{L}_k)$. This gives that $X_k - X_{k-1} = B_{\text{H}}(\textbf{a}_k, \textbf{b}_k, \textbf{x}_k, \textbf{y}_k) - \overline{B}^{SV}_{\text{H}}(\textbf{a}_k, \textbf{b}_k, \textbf{x}_k, \textbf{y}_k) \leq 1$ since $B_{\text{H}}$ is a binary random variable as an indicator. By a standard argument such as shown in \cite{BRGH+16}, we know that $X_0, \dots, X_k$ is a martingale with respect to $U_0, \dots, U_k$. We can now calculate that for all $k = 1, \dots, |\mathcal{S}_{\text{H}}|$, 
\begin{eqnarray}
W_k &=& \sum_{j=1}^{k} \mathbb{E}_{j-1} \left(X_{j} - X_{j-1} \right)^2 \nonumber \\
&=& \sum_{j=1}^{k} \mathbb{E}_{j-1} \left(B_{\text{H}}(\textbf{a}_j, \textbf{b}_j, \textbf{x}_j, \textbf{y}_j) - \overline{B}^{SV}_{\text{H}}(\textbf{a}_j, \textbf{b}_j, \textbf{x}_j, \textbf{y}_j)  \right)^2  \nonumber \\
&=& \sum_{j=1}^{k} \left( \overline{B}^{SV}_{\text{H}}(\textbf{a}_j, \textbf{b}_j, \textbf{x}_j, \textbf{y}_j) - \left( \overline{B}^{SV}_{\text{H}}(\textbf{a}_j, \textbf{b}_j, \textbf{x}_j, \textbf{y}_j) \right)^2 \right) \nonumber \\
& \leq & \sum_{j=1}^k \overline{B}^{SV}_{\text{H}}(\textbf{a}_j, \textbf{b}_j, \textbf{x}_j, \textbf{y}_j) \nonumber \\
&\leq & \sum_{j=1}^{|\mathcal{S}_{\text{H-bad}}|} \overline{B}^{SV}_{\text{H}}(\textbf{a}_j, \textbf{b}_j, \textbf{x}_j, \textbf{y}_j) =: \tilde{\sigma}^2.
\end{eqnarray}
We obtain by applying Freedman's inequality that with probability at least $\text{Pr}\left(W_{|\mathcal{S}_{\text{H}}|} \leq \tilde{\sigma}^2 \right) - \exp  \left\{ \frac{-\beta^2/2}{\tilde{\sigma}^2 +  \beta/3} \right\} = 1 - \exp  \left\{ \frac{-\beta^2/2}{\tilde{\sigma}^2 +  \beta/3} \right\}$, it holds that 
\begin{eqnarray}
\overline{\textit{Z}}_{\text{H}}(A,B,X,Y) \leq \textit{Z}_{\text{H}}(A,B,X,Y) +\frac{\beta}{|\mathcal{S}_{\text{H}}|},
\end{eqnarray}
where 
\begin{eqnarray}
 \overline{\textit{Z}}_{\text{H}}(A,B,X,Y) = \overline{L}_{| \mathcal{S}_{\text{H}}|}.
\end{eqnarray}
Thus, when the first test in Protocol I is passed, i.e., when the event $\texttt{ACC}_1$ occurs, it holds that
\begin{eqnarray}
\overline{\textit{Z}}_{\text{H}}(A,B,X,Y) \leq \delta_1 + \frac{\beta}{|\mathcal{S}_{\text{H}}|} 
\end{eqnarray}
with probability at least $1 - \exp  \left\{ \frac{-\beta^2/2}{|\mathcal{S}_{\text{H}}| \delta_1 +  4\beta/3} \right\}$. We now choose $\beta = |\mathcal{S}_{\text{H}}|^{0.01}$ and $\delta_1 = |\mathcal{S}_{\text{H}}|^{-0.99}$ so that we have
\begin{eqnarray}
 \overline{\textit{Z}}_{\text{H}}(A,B,X,Y) \leq 2  |\mathcal{S}_{\text{H}}|^{-0.99}
\end{eqnarray}
with probability at least $1 - \exp \left\{ -(3/14)  |\mathcal{S}_{\text{H}}|^{0.01} \right\}$.

Now, from 
\begin{eqnarray}
\sum_{i=1}^{|\mathcal{S}_{\text{H}}|}  \overline{B}^{SV}_{\text{H}}(\textbf{a}_i, \textbf{b}_i, \textbf{x}_i, \textbf{y}_i) &\leq & 2 |\mathcal{S}_{\text{H}}|^{0.01} 
\end{eqnarray}
we obtain that
\begin{eqnarray}
p_{\min}(\epsilon, N) \sum_{i=1}^{|\mathcal{S}_{\text{H}}|} \mathcal{I}_0\left(P_{\textbf{A}_i, \textbf{B}_i | \textbf{X}_i, \textbf{Y}_i} \right) &\leq & 2 |\mathcal{S}_{\text{H}}|^{0.01}  \nonumber \\
p_{\min}(\epsilon, N) \left[\sum_{i \in \mathcal{S}_{\text{H-bad}} }  \mathcal{I}_0\left(P_{\textbf{A}_i, \textbf{B}_i | \textbf{X}_i, \textbf{Y}_i} \right) + \sum_{i \notin \mathcal{S}_{\text{H-bad}} } \mathcal{I}_0\left(P_{\textbf{A}_i, \textbf{B}_i | \textbf{X}_i, \textbf{Y}_i} \right) \right] &\leq & 2 |\mathcal{S}_{\text{H}}|^{0.01}  \nonumber \\
p_{\min}(\epsilon, N) |\mathcal{S}_{\text{H-bad}}| \kappa &\leq & 2 |\mathcal{S}_{\text{H}}|^{0.01},
\end{eqnarray}
so that the number of bad runs (in which the value of $\mathcal{I}_0$ is greater than $\kappa$) is bounded as 
\begin{eqnarray}
|\mathcal{S}_{\text{H-bad}}| \leq \frac{2 |\mathcal{S}_{\text{H}}|^{0.01}}{p_{\min}(\epsilon, N)\kappa} \leq \frac{2}{ p_{\min}(\epsilon, N) \kappa} \left(\frac{3M}{N+1} \right)^{0.01}.
\end{eqnarray}
Here $p_{\min}(\epsilon, N)$ denotes the minimum probability with which the measurement settings are chosen in any particular run, i.e., 
\begin{eqnarray}
\nu^{SV}_{\textbf{X}_i, \textbf{Y}_i}(\textbf{x}_i, \textbf{y}_i) \geq p_{\min}(\epsilon, N), \; \; \; \; \forall \textbf{x}_i, \textbf{y}_i, i \in \mathcal{S}_{\text{H}}.
\end{eqnarray} 
Since the settings are chosen with an $\epsilon$-SV source, and the parties discard the runs in which the $2(N+1)$ settings that are involved in the Hardy ladder test do not appear, one can bound $p_{\min}(\epsilon, N)$ as \cite{CR12}
\begin{eqnarray}
p_{\min}(\epsilon, N) \geq \frac{\left(\frac{1}{2} - \epsilon \right)^{2 \log (N+1)}}{2(N+1) \left(\frac{1}{2} + \epsilon \right)^{2 \log (N+1)}} = \frac{1}{2} (N+1)^{2 \log (1/2 - \epsilon) - 2 \log (1/2 + \epsilon) - 1}.
\end{eqnarray}
Setting $\kappa = (N+1)^{-t}$ for a real parameter $t > 0$, and choosing the number of runs as $M = (N+1)^r$ gives that the number of bad runs (conditioned upon the first test being accepted) is bounded as
\begin{eqnarray}
\label{eq:testone-bad-bound}
|\mathcal{S}_{\text{H-bad}}| \leq 4 (3)^{0.01} (N+1)^{0.99 + t + 0.01 r +2 \log (1/2+ \epsilon) - 2 \log (1/2 - \epsilon)} =: 4 (3)^{0.01} (N+1)^{t'}, 
\end{eqnarray}
with probability at least $1 - \exp \left\{ -(3/14) (N+1)^{0.01(r-1)} \right\}$, which is close to $1$ for $ r > 1$. 

In the Protocol I, the honest parties perform a second tomographic test to estimate the probability $P_{\textbf{A}, \textbf{B}| \textbf{X}, \textbf{Y}}(0,0|N,N)$ in the runs in which the measurement setting $(N,N)$ appears (i.e., within the set $\mathcal{S}_{\text{H-nz}}$). This test checks that
\begin{eqnarray}
Z_{\text{H-nz}}(A,B,X,Y) := \frac{1}{|\mathcal{S}_{\text{H-nz}}|} \sum_{i \in \mathcal{S}_{\text{H-nz}}} B_{\text{H-nz}}(\textbf{a}_i, \textbf{b}_i, \textbf{x}_i, \textbf{y}_i) \geq \frac{1}{2} - \delta_2,
\end{eqnarray}
for fixed $\delta_2 > 0$, where recall that $\mathcal{S}_{\text{H-nz}}$ denotes the set of runs in which the setting $(N,N)$ appears, and $B_{\text{H-nz}}(\textbf{a}_i, \textbf{b}_i, \textbf{x}_i, \textbf{y}_i)$ is defined from Eq.(\ref{eq:B-Hardy-nz}) as
\begin{equation}
B_{\text{H-nz}}(\textbf{a}_i, \textbf{b}_i, \textbf{x}_i, \textbf{y}_i)=  \begin{cases} 
      1 & (\textbf{x}_i, \textbf{y}_i) = (N,N) \wedge (\textbf{a}_i, \textbf{b}_i) = (0,0)  \\
      0 & \text{otherwise}
   \end{cases}
\end{equation}
This indicator function takes value $1$ whenever the measurement outcome $(\textbf{a}_i, \textbf{b}_i) = (0,0)$ appears for a run in which setting $(\textbf{x}_i, \textbf{y}_i)  = (N, N)$ was chosen, and takes value $0$ otherwise. In the ideal noiseless scenario, the parties observe the value of $Z_{\text{H-nz}}(A,B,X,Y)$ close to $1/2$, with an expected deviation from $1/2$ of $N^{-d}$ with $d = 0.99$.


By computing the empirical average $Z_{\text{H-nz}}(A,B,X,Y)$, the parties estimate the value of the true average $\overline{Z}_{\text{H-nz}}(A,B,X,Y)$ defined as
\begin{eqnarray}
\overline{Z}_{\text{H-nz}}(A,B,X,Y) := \frac{1}{|\mathcal{S}_{\text{H-nz}}|} \sum_{i \in \mathcal{S}_{\text{H-nz}}} P_{\textbf{A}_i, \textbf{B}_i| \textbf{X}_i, \textbf{Y}_i}(0,0| N, N). \nonumber \\
\end{eqnarray}
In the estimation procedure, we employ the following Lemma \ref{lemmaazuma} shown in \cite{BRGH+16} based on the Azuma-Hoeffding inequality. 
\begin{lemma}[\cite{BRGH+16}] 
	\label{lemmaazuma} 
	Consider  arbitrary random variables $W_i$ for $i=0,1,\ldots,n$, and binary random variables $B_i$ for $i=1,\ldots n$ that are functions of 
	$W_i$, i.e. $B_i=f_i(W_i)$ 
	for some functions $f_i$. Let us denote $\overline{B}_i=\mathbb{E}(B_i|W_{i-1},\ldots,W_1,W_0)$ for $i=1,\ldots,n$ 
	(i.e. $\overline{B}_i$ are conditional means).
	Define for $k = 1, \ldots, n$, the empirical average
	\begin{eqnarray}
	L_k=\frac{1}{k}\sum_{i=1}^k B_i
    \label{eq:means}
	\end{eqnarray}
	and the arithmetic average of conditional means 
	\begin{eqnarray}
	\label{eq:cond-means}
	\overline{L}_k=\frac{1}{k}\sum_{i=1}^k \overline{B}_i.
	\end{eqnarray}
	Then we have 
	\begin{eqnarray}
	\text{Pr}(|L_n-\overline{L}_n|\geq  \delta_{Az})\leq 2 e^{-n\frac{\delta_{Az}^2}{2}}
	\label{eq:PrAzumaL-Ln}
	\end{eqnarray}
\end{lemma}
By Lemma \ref{lemmaazuma}, we see that when $Z_{\text{H-nz}}(A,B,X,Y) \geq \frac{1}{2} - \delta_2$, with probability at least $1 - 2  e^{- |\mathcal{S}_{\text{H-nz}}| \frac{\delta_{Az}^2}{2}}$ it holds that
\begin{eqnarray}
\overline{Z}_{\text{H-nz}}(A,B,X,Y) \geq \frac{1}{2} - \delta_2 - \delta_{Az},
\end{eqnarray}
for any fixed $\delta_{Az} > 0$. Note that in the protocol, the parties verify that $\frac{M}{2(N+1)^2} \leq |\mathcal{S}_{\text{H-nz}}| \leq \frac{3M}{2(N+1)^2}$, so the statement holds with probability at least $1 - 2  \exp \big\{- \frac{M \delta_{Az}^2}{4(N+1)^2} \big\}$. For $\delta_{Az} = (N+1)^{-\gamma}$ and $M = (N+1)^r$, this probability is $1 - 2 \exp \big\{ - \frac{1}{4}(N+1)^{r-2 - 2 \gamma} \big\}$ which is close to unity for $r > 2 + 2 \gamma$. 

The following Lemma is then useful to estimate the true value $P_{\textbf{A}_i, \textbf{B}_i | \textbf{X}_i, \textbf{Y}_i}(0,0|N,N)$ for a large fraction of the boxes within the set $\mathcal{S}_{\text{H-nz}}$. 

\begin{lemma}
If the average $\overline{Z}_{\text{H-nz}}(A,B,X,Y)$ satisfies $\overline{Z}_{\text{H-nz}}(A,B,X,Y) \geq \frac{1}{2} - \delta_2 - \delta_{Az}$ and when the number of runs $|\mathcal{S}_{\text{H-bad}}|$ in which $\mathcal{I}_0 \geq (N+1)^{-t}$ is bounded as $| \mathcal{S}_{\text{H-bad}}| \leq 4 (3)^{0.01} (N+1)^{t'}$ from Eq.\eqref{eq:testone-bad-bound}, then in at most $\frac{ \big|\mathcal{S}_{\text{H-nz}}\big| \left[\frac{1}{2}(N+1)^{-t} + \delta_2 + \delta_{Az} \right] + 4(3)^{0.01} (N+1)^{t'}}{\frac{1}{2}(N+1)^{-t}+ \delta_3}$ runs $i \in \mathcal{S}_{\text{H-nz}}$, we have $P_{\textbf{A}_i, \textbf{B}_i | \textbf{X}_i, \textbf{Y}_i}(0,0|N,N) \leq \frac{1}{2} - \delta_3$ for fixed constant $\delta_3 > 0$. 
\end{lemma}
\begin{proof}
We have that $\overline{Z}_{\text{H-nz}}(A,B,X,Y) \geq \frac{1}{2} - \delta_2 - \delta_{Az}$, i.e., 
\begin{eqnarray}
\sum_{i \in \mathcal{S}_{\text{H-nz}}} P_{\textbf{A}_i, \textbf{B}_i | \textbf{X}_i, \textbf{Y}_i}(0,0| N,N) \geq \big|\mathcal{S}_{\text{H-nz}}\big| \left(\frac{1}{2} - \delta_2 - \delta_{Az} \right).
\end{eqnarray}
Let us define the subset $\mathcal{S}_{\text{H-nz-bad}} \subset \mathcal{S}_{\text{H-nz}}$ of runs $i \in \mathcal{S}_{\text{H-nz}}$ in which the value of $P_{\textbf{A}_i, \textbf{B}_i | \textbf{X}_i, \textbf{Y}_i}(0,0|N, N)$ falls below a specified threshold $\frac{1}{2} - \delta_3$, i.e., we define
\begin{eqnarray}
\label{eq:def-bad-nz}
\mathcal{S}_{\text{H-nz-bad}} := \bigg\{ i \in \mathcal{S}_{\text{H-nz}} \bigg| P_{\textbf{A}_i, \textbf{B}_i | \textbf{X}_i, \textbf{Y}_i}(0,0| N, N) \leq \frac{1}{2} - \delta_3 \bigg\},
\end{eqnarray}
for fixed $\delta_3 > 0$. We have
\begin{eqnarray}
\label{eq:S11-size}
 \big|\mathcal{S}_{\text{H-nz}}\big| \left(\frac{1}{2} - \delta_2 - \delta_{Az} \right) &\leq & \sum_{i \in \mathcal{S}_{\text{H-nz}}} P_{\textbf{A}_i, \textbf{B}_i | \textbf{X}_i, \textbf{Y}_i}(0,0| N,N) \nonumber \\ &=& \sum_{i \in \mathcal{S}_{\text{H-nz-bad}}} P_{\textbf{A}_i, \textbf{B}_i | \textbf{X}_i, \textbf{Y}_i}(0,0| N,N) + \sum_{i \notin \mathcal{S}_{\text{H-nz-bad}}} P_{\textbf{A}_i, \textbf{B}_i | \textbf{X}_i, \textbf{Y}_i}(0,0|N, N) \nonumber \\
 & \leq & \big| \mathcal{S}_{\text{H-nz-bad}} \big| \left(\frac{1}{2} - \delta_3 \right) + \sum_{i \in  \mathcal{S}^{c}_{\text{H-nz-bad}}} P_{\textbf{A}_i, \textbf{B}_i | \textbf{X}_i, \textbf{Y}_i}(0,0| N,N).
\end{eqnarray}
where $\mathcal{S}^{c}_{\text{H-nz-bad}}$ denotes the complement of $\mathcal{S}_{\text{H-nz-bad}}$ within $\mathcal{S}_{\text{H-nz}}$. 

Now, we know that conditioned upon the event $\texttt{ACC}_1)$, we have the number of bad runs in which $\mathcal{I}_0\left( P_{\textbf{A}_j, \textbf{B}_j | \textbf{X}_j, \textbf{Y}_j} \right) > (N+1)^{-t}$ is bounded as $| \mathcal{S}_{\text{H-bad}}| \leq 4 (3)^{0.01}(N+1)^{t'}$ from Eq.\eqref{eq:testone-bad-bound}.
Furthermore, we have by Lemma \ref{lem:single-run-bound}, that when $\mathcal{I}_0\left(P_{\textbf{A}_i, \textbf{B}_i| \textbf{X}_i, \textbf{Y}_i} \right)  \leq (N+1)^{-t}$ that
\begin{eqnarray}
\label{eq:single-bound-2}
P_{\textbf{A}_i, \textbf{B}_i| \textbf{X}_i, \textbf{Y}_i}(0,0|N, N) \leq \frac{1}{2} \bigg( 1 + \mathcal{I}_0\left(P_{\textbf{A}_i, \textbf{B}_i| \textbf{X}_i, \textbf{Y}_i} \right) \bigg) \leq \frac{1}{2} \left[1 + (N+1)^{-t}\right].
\end{eqnarray}
Therefore, we have that conditioned  upon $\texttt{ACC}_1$, in at least $ \big|\mathcal{S}_{\text{H-nz}}\big| - 4(3)^{0.01}(N+1)^{t'}$ runs, it holds that $P_{\textbf{A}_i, \textbf{B}_i| \textbf{X}_i, \textbf{Y}_i}(0,0|N, N)$ is bounded as in \eqref{eq:single-bound-2}.
Plugging this into \eqref{eq:S11-size} gives
\begin{eqnarray}
 \big|\mathcal{S}_{\text{H-nz}}\big|  \left(\frac{1}{2} - \delta_2 - \delta_{Az} \right) &\leq &  \big| \mathcal{S}_{\text{H-nz-bad}} \big| \left(\frac{1}{2} - \delta_3 \right) \nonumber \\ && +\sum_{i \in \left(\mathcal{S}_{\text{H-nz}} \setminus \mathcal{S}_{\text{H-nz-bad}} \right) \cap \mathcal{S}_{\text{H-bad}}} P_{\textbf{A}_i, \textbf{B}_i | \textbf{X}_i, \textbf{Y}_i}(0,0| N, N) \nonumber \\
&&+ \sum_{i \in \left(\mathcal{S}_{\text{H-nz}} \setminus \mathcal{S}_{\text{H-nz-bad}} \right) \cap \left( \mathcal{S}_{\text{H}} \setminus \mathcal{S}_{\text{H-bad}} \right)} P_{\textbf{A}_i, \textbf{B}_i | \textbf{X}_i, \textbf{Y}_i}(0,0| N, N) \nonumber \\
 & \leq &  \big| \mathcal{S}_{\text{H-nz-bad}} \big| \left(\frac{1}{2} - \delta_3 \right) + \big| \mathcal{S}_{\text{H-bad}} \big| + \bigg| \left(\mathcal{S}_{\text{H-nz}} \setminus \mathcal{S}_{\text{H-nz-bad}} \right) \bigg| \frac{1}{2} \left[1 + (N+1)^{-t} \right] \nonumber \\
  & \leq & \big| \mathcal{S}_{\text{H-nz-bad}} \big| \left(\frac{1}{2} - \delta_3 \right) + 4(3)^{0.01}(N+1)^{t'} +\bigg(  \big|\mathcal{S}_{\text{H-nz}}\big|- \big| \mathcal{S}_{\text{H-nz-bad}} \big| \bigg) \frac{1}{2} \left[1 +  (N+1)^{-t} \right]. \nonumber \\
\end{eqnarray}
From the above we obtain that
\begin{eqnarray}
\big| \mathcal{S}_{\text{H-nz-bad}}  \big| \left[\frac{1}{2}(N+1)^{-t}+ \delta_3 \right] \leq \big|\mathcal{S}_{\text{H-nz}}\big| \left[\frac{1}{2}(N+1)^{-t} + \delta_2 + \delta_{Az} \right] + 4(3)^{0.01} (N+1)^{t'}.
\end{eqnarray}
\end{proof}
 
We note that in the ideal noiseless scenario, the test parameter 
$\delta_2$ is of the order of $N^{-d}$ for parameter $d < 1$, we set $d = 0.99$. By suitably defining the value of $\delta_3$, we can ensure that the size of the subset $\mathcal{S}_{\text{H-nz-bad}}$ is much smaller than that of $\mathcal{S}_{\text{H-nz}}$. Furthermore, we have that the set of bad runs from the first test is bounded as $| \mathcal{S}_{\text{H-bad}}| \leq  4(3)^{0.01}(N+1)^{t'}$ from Eq.\eqref{eq:testone-bad-bound} with high probability when the first test is passed. Therefore, we can define the total set of bad runs $\mathcal{S}_{\text{bad}}$ as 
\begin{eqnarray}
\label{eq:def-bad}
\mathcal{S}_{\text{bad}} & := & \mathcal{S}_{\text{H-bad}}  \cup \mathcal{S}_{\text{H-nz-bad}} \nonumber \\
&=&  \bigg\{ i \in \mathcal{S}_{\text{H-nz}} \bigg| P_{\textbf{A}_i, \textbf{B}_i | \textbf{X}_i, \textbf{Y}_i}(0,0| N, N) \leq \frac{1}{2} - \delta_3 \vee \mathcal{I}_0 \left( P_{\textbf{A}_i, \textbf{B}_i | \textbf{X}_i, \textbf{Y}_i} \right) \geq  (N+1)^{-t} \bigg\},
\end{eqnarray}
and obtain that for chosen parameters $t = 0.99$, $\delta_2 = \delta_{Az} = (N+1)^{-0.99}$, $\delta_3 = (N+1)^{-0.01}$, $M = (N+1)^r$ that with probability greater than $1 - 2  \exp \left\{- \frac{1}{4} (N+1)^{(r-3.98)}\right\} - \exp \left\{ -\frac{3}{14} (N+1)^{0.01(r-1)} \right\}$
\begin{eqnarray}
|\mathcal{S}_{\text{bad}}| \leq \big| \mathcal{S}_{\text{H-bad}}  \big|  + \big| \mathcal{S}_{\text{H-nz-bad}} \big| &\leq &4(3)^{0.01}(N+1)^{t'} + \frac{1}{\delta_3}  \big|\mathcal{S}_{\text{H-nz}}\big| \left[\frac{1}{2}(N+1)^{-t} + \delta_2 + \delta_{Az} \right] + \frac{4(3)^{0.01} (N+1)^{t'}}{\delta_3}  \nonumber \\
&\leq & 4(3)^{0.01}(N+1)^{t'} +  \frac{5 |\mathcal{S}_{\text{H-nz}}| (N+1)^{-0.98}}{2} + 4(3)^{0.01}(N+1)^{t'+0.01}   \nonumber \\
&\leq & \frac{15 (N+1)^{r-2.98}}{4} + 9(N+1)^{t'+0.01},
\end{eqnarray}
where $t' = 0.01 r + 1.98 + 2 \log(1/2+ \epsilon) - 2 \log (1/2 - \epsilon)$, where we used the fact that $|\mathcal{S}_{\text{H-nz}}| \leq \frac{3M}{2(N+1)^2} = \frac{3}{2}(N+1)^{r-2}$.
\subsection{Security of the final output bit for a range of initial $\epsilon$.}
We now bound the probability that the box $r \in  \mathcal{S}_{\text{H-nz}}$ chosen using $\log |\mathcal{S}_{\text{H-nz}}|$ bits from the $\epsilon$-SV source belongs to the bad set $\mathcal{S}_{\text{bad}}$. Since the probability of choosing one of the runs from $\mathcal{S}_{\text{H-nz}}$ using an $\epsilon$-SV source is at most $\left(1/2 + \epsilon \right)^{\log |\mathcal{S}_{\text{H-nz}}|}$, we have by the union bound that 
\begin{eqnarray}
\text{Pr}\left(r \in \mathcal{S}_{\text{bad}} \right)& \leq & \left(1/2 + \epsilon \right)^{\log |\mathcal{S}_{\text{H-nz}}|} \big|\mathcal{S}_{\text{bad}} \big| \nonumber \\
& \leq & |\mathcal{S}_{\text{H-nz}}|^{\log (1/2 + \epsilon)} \left[ \frac{15 (N+1)^{r-2.98}}{4} + 9(N+1)^{t'+0.01} \right] \nonumber \\
& \leq & \left(\frac{3(N+1)^{r-2}}{2} \right)^{\log (1/2+ \epsilon)} \left[ \frac{15 (N+1)^{r-2.98}}{4} +9(N+1)^{t'+0.01} \right] \nonumber \\
& \leq & c \left[(N+1)^{r[1 + \log (1/2+\epsilon)] - 2 \log (1/2+ \epsilon) - 2.98} + (N+1)^{(r-2) \log (1/2+\epsilon) + t' + 0.01} \right], 
\end{eqnarray}  
where $c$ is a constant $c = (3/2)^{6 + \log (1/2+ \epsilon)}$. 
We obtain that the probability $\text{Pr}\left(r \in \mathcal{S}_{\text{bad}} \right) \rightarrow 0$ whenever $\epsilon \lessapprox 0.090$ for the choice 
\begin{eqnarray}
r = \frac{2.98 + 2 \log (1/2 + \epsilon)}{1 + \log (1/2 + \epsilon)} - 0.01.
\end{eqnarray}

In the set of good runs $\mathcal{S}_{\text{good}} = \mathcal{S}^{c}_{\text{bad}}$ which is the complement of Eq.\eqref{eq:def-bad}, the probability $P_{\textbf{A}_i, \textbf{B}_i | \textbf{X}_i, \textbf{Y}_i, W = w, Z =z}(0,0|N,N)$ is bounded as 
\begin{eqnarray}
\frac{1}{2} - \delta_3 \leq P_{\textbf{A}_i, \textbf{B}_i | \textbf{X}_i, \textbf{Y}_i, W = w, Z =z}(0,0|N,N) \leq \frac{1}{2}\left[1 + (N+1)^{-t} \right].
\end{eqnarray}
We can now estimate the distance of the final bit to uniform $D\left(P_{R_r, W, Z|N,N}, P_{U} \times P_{W, Z| N, N}\right)$. To do this, we use the reasoning from the proof of Prop. \ref{prop:quant-assum}. Specifically from Eq. \eqref{eq:dist-final-bit}, we obtain that $D\left(P_{R_r, W, Z|N,N}, P_{U} \times P_{W, Z| N, N}\right)$ is bounded as
\begin{eqnarray}
D\left(P_{R_r, W, Z|N,N}, P_{U} \times P_{W, Z| N, N}\right) &\leq & \max \bigg\{ \delta_3, (N+1)^{-t} \left( \frac{1/2 +  \epsilon}{1/2 -  \epsilon} \right)^{2 \log (N+1)} \bigg\} \nonumber \\
&=& \max \bigg\{ (N+1)^{-0.01}, (N+1)^{-0.99}  \left( \frac{1/2 +  \epsilon}{1/2 -  \epsilon} \right)^{2 \log (N+1)} \bigg\} \nonumber \\
&=& \max \bigg\{ (N+1)^{-0.01}, (N+1)^{-0.99 + 2 \log (1/2 + \epsilon) - 2 \log (1/2 - \epsilon)} \bigg\}.
\end{eqnarray}
This gives that, conditioned on $\texttt{ACC} = \texttt{ACC}_1 \cap \texttt{ACC}_2$, the final output bit is arbitrarily close to uniform when for all $\epsilon \lessapprox 0.085$ when $-0.99 + 2 \log (1/2 + \epsilon) - 2 \log (1/2 - \epsilon) \leq 0$. We also check that for this range of $\epsilon$, it holds that $r > 3.98$ so that the probabilities in the concentration inequalities are close to unity. 

The correctness of the Protocol then follows from the value achieved by quantum correlations for both tests, since both tests in the protocol can be passed with high probability (i.e., event $\texttt{ACC}$ occurs with high probability) when the parties use honest devices that share the appropriate quantum correlations. 

We remark that the above value of $\epsilon$ is not necessarily optimal and is an artefact of the method of our proof. Given Prop. \ref{prop:quant-assum}, one may well expect that an improved value of $\epsilon$ may be obtained from a different analysis, or a different protocol based on the same ladder Hardy inequality. We further remark that it is also possible to modify the protocol to one such as in \cite{RBHH+16} where the requirement that $M/2(N+1)^2 \leq \big| \mathcal{S}_{\text{H-nz}} \big| \leq 3M/2(N+1)^2$ is dropped. In its place, one could then apply a generalized Chernoff bound to deduce that with high probability, the input $(\textbf{x}_i, \textbf{y}_i) = (N,N)$ appears for a linear fraction of the runs in the Protocol \cite{RBHH+16}.

\section{Supplemental Material: Sketch of proof of security of one-sided device-independent randomness amplification}
In this section, we sketch the proof of security of the one-sided device-independent (1sDI) randomness amplification protocol from the main text. We defer a full security proof to future work for the following reason. It may be expected that this protocol can be modified to allow for a constant rate of noise. However, such a modification would require the application of a (strong) randomness extractor at the end of the protocol rather than obtaining the output bit from the outcomes of one randomly chosen run. It is then required to study if the techniques of the proofs against quantum side information from \cite{ADF+18, KAF17} can be adapted to prove general security in the 1sDI scenario. 

In this section, we show the existence of extremal sequential non-signalling assemblages that admit quantum realization in the time-ordered no-signaling (TONS) scenario. Furthermore, we show that such extremal assemblages are such that the outcomes of each party are perfectly random with respect to any non-signalling adversary. As such, under the correctness assumption that the devices behave according to quantum theory, picking the output bit from the measurement outcomes in a randomly chosen run (using an $\epsilon$-SV source) guarantees that the output bit is uniformly random and private with respect to any non-signalling adversary. Furthermore, since quantum correlations achieve algebraic violation of the steering inequality used in this protocol, the certification of perfect randomness can be made for arbitrary measurement independence, i.e., for arbitrary initial $\epsilon < 0.5$.

\subsection{No-signaling assemblages - Preliminaries}

In this subsection we give some preliminaries on the notions of no-signaling assemblages used to characterise the devices held by the honest parties and adversary Eve in the one-sided device-independent (1sDI) randomness amplification protocol. 

We consider a four-party device with three honest users Alice (A), Bob (B), Charlie (C) and an adversary Eve (E). The reduced three-party system held by the honest parties consists of two uncharacterised subsystems held by A-B and a trusted characterised system held by C. The subsystem C is quantum and is described by some Hilbert space of known dimension $d_C$, while the subsystems A and B are black box devices described only by a set of classical inputs (labels of the measurement settings) and classical outputs (labels of the measurement outcomes). In each run of the protocol, Alice and Bob perform local measurements $x, y$ respectively chosen using an $\epsilon$-SV source and obtain outcomes $a,b$ respectively.

When Alice and Bob choose the settings $x, y$, the quantum system of Charlie is described by the normalised version of operator $\sigma^{(C)}_{ab|xy}$ with probability $\mathrm{Tr}(\sigma^{(C)}_{ab|xy})$ (which describes the probability of Alice and Bob obtaining outcomes labeled by $a,b$). The  operator satisfies the no-signaling, non-negativity and normalisation conditions given by
\begin{equation}\label{sup_as1}
\forall_{a,b,x,y}\ \sigma^{(C)}_{ab|xy}\geq 0,
\end{equation}
\begin{equation}\label{sup_as2}
\forall_{b,x,x',y}\ \sum_a\sigma^{(C)}_{ab|xy}=\sum_a\sigma^{(C)}_{ab|x'y},
\end{equation}
\begin{equation}\label{sup_as3}
\forall_{a,x,y,y'}\ \sum_b\sigma^{(C)}_{ab|xy}=\sum_b\sigma^{(C)}_{ab|xy'},
\end{equation}
\begin{equation}\label{sup_as4}
\forall_{x,y}\ \mathrm{Tr}(\sum_{a,b}\sigma^{(C)}_{ab|xy})=\mathrm{Tr}(\sigma^{(C)})=1.
\end{equation}
The abstract tripartite no-signaling assemblage (with two uncharacterised subsystems) is then defined as a collection of operators $\Sigma^{(C)}=\left\{\sigma^{(C)}_{ab|xy}\right\}_{a,b,x,y}$ acting on fixed $d_C$ dimensional space, which satisfy constraints (\ref{sup_as1}-\ref{sup_as4}) \cite{SBCSV15}.  
If for a given no-signaling assemblage $\left\{\sigma^{(C)}_{ab|xy}\right\}_{a,b,x,y}$ there exist subsystems A and B, local POVMs given respectively by elements $M^{(A)}_{a|x},N^{(B)}_{b|y}$ and some tripartite state $\rho^{(ABC)}$ of the composite system ABC such that
\begin{equation}\label{sup_q2}
\sigma^{(C)}_{ab|xy}=\mathrm{Tr}_{AB}\left(M^{(A)}_{a|x}\otimes N^{(B)}_{b|y}\otimes \mathds{1}\rho^{(ABC)} \right),
\end{equation}
we say that $\Sigma^{(C)}$ admits a quantum realisation (or simply that $\Sigma^{(C)}$ is a quantum assemblage).

It is clear that the set of no-signaling assemblages and the subsets of quantum assemblages (assemblages with quantum realisation) are convex. The same is true for the subset of local hidden state (LHS) assemblages - a no-signaling assemblage from tripartite steering admits an LHS model \cite{SAPHS18} if it can be represented as 
\begin{equation}\label{sup_lhs2}
\sigma^{(C)}_{ab|xy}=\sum_i q_i p^{(A)}_i(a|x)p^{(B)}_i(b|y)\sigma^{(C)}_i
\end{equation}where $q_i\geq 0, \sum_i q_i=1$, $\sigma^{(C)}_i$ are some states of the characterised subsystem C and $p^{(A)}_i(a|x),p^{(B)}_i(b|y)$ denote conditional probability distributions for the uncharacterised subsystems A and B respectively.

\subsection{Sequential no-signaling assemblages}\label{sec1}

Consider a tripartite setting ABC in which each separated party obtains a particle at time $t_i$ where $i=1,\ldots, n$. Let two subsystems (A,B) be uncharacterised and let one subsystem be characterised (C). For the characterised subsystem C, let any particle obtained at $t_i$ be described by a Hilbert space $H_{C_i}$ and let then the whole subsystem C be related to $\otimes_{i=1}^{n}H_{C_{i}}$. Finally, take analogues description of particles arriving at subsystems A and B. Consider ordered sequence of points in time $t_1<t_2<\ldots <t_n$ and a steering scenario in which $A$ performs at $t_i$ local uncharacterised measurement (on the particles which arrived not later than $t_i$) with setting $x_i$ and outcome $a_i$, $B$ performs at $t_i$ local uncharacterised measurement (on the particles which arrived not later than $t_i$) with setting $y_i$ and outcome $b_i$, and there is no signaling between A and B as well as there is no signaling backward in time (from $t_i$ to $t_j$ with $j<i$), i.e. outcomes of past measurements do not depend on future settings.

For simplicity we will sometimes write $\textbf{a}$ instead of ordered n-tuple of outcomes $(a_1,\ldots, a_n)$, and for a given $0\leq k\leq n$ we will use notation $\textbf{a}_{\leq k}$ for the tuple obtained from $\textbf{a}$ by discarding outcomes $a_{i}$ with $i>k$. Finally, we will put $\textbf{a}_{\overline{k}}=(a_1,\ldots,a_{k-1},a_{k+1},\ldots, a_n)$. When convenient, we will use the same convention for $\textbf{i}=(1,\ldots, n)$ and any n-tuple related respectively to $b_i,x_i$ or $y_i$.

Using this notation we can describe quantum assemblage related to the discussed sequential scenario. \textit{Sequential quantum assemblage} is a collection of subnormalised states $\sigma_{\textbf{a}\textbf{b}|\textbf{x}\textbf{y}}$ obtained by
\begin{equation}\label{ass}
\sigma_{\textbf{a}\textbf{b}|\textbf{x}\textbf{y}}:=\mathrm{Tr}_{AB}\left(K^{(n)\dagger}_{a_n|x_{n}}\otimes L^{(n)\dagger}_{b_n|y_{n}}\ldots K^{(1)\dagger}_{a_1|x_{1}}\otimes L^{(1)\dagger}_{b_1|y_{1}}\rho_{ABC}K^{(1)}_{a_1|x_{1}}\otimes L^{(1)}_{b_1|y_{1}}\ldots K^{(n)}_{a_n|x_{n}}\otimes L^{(n)}_{b_n|y_{n}}\right)
\end{equation}where $\rho_{ABC}\in \otimes_{i=1}^{n}B(H_{A_i}\otimes H_{B_i}\otimes H_{C_{i}})$ is a given state and for any $1\leq k\leq n$ and any fixed $x_{k}, y_{k}$, operators $K^{(k)}_{a_k|x_{k}}K^{(k)\dagger}_{a_k|x_{ k}}\in \otimes_{i=1}^{k}B(H_{A_i})$ and $L^{(k)}_{b_k|y_{ k}}L^{(k)\dagger}_{b_k|y_{k}}\in\otimes_{i=1}^{k}B(H_{B_i})$ are elements of POVMs, i.e. $\sum_{a_k}K^{(k)}_{a_k|x_{k}}K^{(k)\dagger}_{a_k|x_{k}}=\mathds{1}\in \otimes_{i=1}^{k}B(H_{A_i})$ and $L^{(k)}_{b_k|y_{ k}}L^{(k)\dagger}_{b_k|y_{k}}=\mathds{1}\in \otimes_{i=1}^{k}B(H_{B_i})$. Note that we abuse the notation by writing $\mathds{1}$ for identity operator acting on a Hilbert space of arbitrary dimension, depending on the context. Moreover, in the above definition of $\sigma_{\textbf{a}\textbf{b}|\textbf{x}\textbf{y}}$, we put $K^{(n)\dagger}_{a_n|x_{n}}\otimes L^{(n)\dagger}_{b_n|y_{n}}\in \otimes_{i=1}^{k}B(H_{A_i}\otimes H_{B_i})$ instead of $K^{(n)\dagger}_{a_n|x_{n}}\otimes L^{(n)\dagger}_{b_n|y_{n}}\otimes \mathds{1}\in \otimes_{i=1}^{n}B(H_{A_i}\otimes H_{B_i}\otimes H_{C_{i}})$ for the sake of notation compactness.

If we forget about quantum nature of measurements on $A, B$, we may try to generalized this picture in the language of no-signaling assemblages. In that case description of the steering scenario in a sequential setting will by given by the \textit{sequential no-signaling assemblage}, i.e. by the set   
\begin{equation}\label{ass}
\Sigma^{(n)}=\left\{\sigma_{\textbf{a}\textbf{b}|\textbf{x}\textbf{y}}\right\}_{\textbf{a}\textbf{b}|\textbf{x}\textbf{y}}
\end{equation}of positive operators (subnormalised states) $\sigma_{\textbf{a}\textbf{b}|\textbf{x}\textbf{y}}=\sigma_{a_1b_1\ldots a_{n}b_{n}|x_1y_1\ldots x_{n}y_{n}}\in \otimes_{i=1}^{n}B(H_{C_{i}})$ satisfying the following conditions. \newline

\noindent 1) For any choice of $\textbf{x},\textbf{y}$
\begin{equation}\label{0con1}
\sum_{a_1,b_1,\ldots, a_n,b_n}\sigma_{a_1b_1\ldots a_nb_n|x_1y_1\ldots x_ny_n}=\sigma, \ \mathrm{Tr}(\sigma)=1.
\end{equation}

\noindent 2) For any $0\leq k\leq n$ and any choice of $\textbf{a}_{\leq k},\textbf{b},\textbf{x},\textbf{y}$ 
\begin{equation}\label{0con2}
\sigma_{\textbf{a}_{\leq k}\textbf{b}|\textbf{x}\textbf{y}}=\sigma_{\textbf{a}_{\leq k}\textbf{b}|\textbf{x}_{\leq k}\textbf{y}}
\end{equation}where 
\begin{equation}
\sigma_{\textbf{a}_{\leq k}\textbf{b}|\textbf{x}\textbf{y}}:=\sum_{a_{k+1},\ldots, a_n}\sigma_{a_1b_1\ldots a_nb_n|x_1y_1\ldots x_ny_n}.
\end{equation}
\noindent 3) For any $0\leq k\leq n$ and any choice of $\textbf{a},\textbf{b}_{\leq k},\textbf{x},\textbf{y}$ 
\begin{equation}\label{0con3}
\sigma_{\textbf{a}\textbf{b}_{\leq k}|\textbf{x}\textbf{y}}=\sigma_{\textbf{a}\textbf{b}_{\leq k}|\textbf{x}\textbf{y}_{\leq k}}
\end{equation}where 
\begin{equation}
\sigma_{\textbf{a}\textbf{b}_{\leq k}|\textbf{x}\textbf{y}}:=\sum_{b_{k+1},\ldots, b_n}\sigma_{a_1b_1\ldots a_nb_n|x_1y_1\ldots x_ny_n}.
\end{equation}

It is useful to present elements of sequential no-signaling assemblage as $\sigma_{\textbf{a}\textbf{b}|\textbf{x}\textbf{y}}=p_{\textbf{a}\textbf{b}|\textbf{x}\textbf{y}}\rho_{\textbf{a}\textbf{b}|\textbf{x}\textbf{y}}$ where $p_{\textbf{a}\textbf{b}|\textbf{x}\textbf{y}}=\mathrm{Tr}(\sigma_{\textbf{a}\textbf{b}|\textbf{x}\textbf{y}})$ stand for conditional probabilities and each $\rho_{\textbf{a}\textbf{b}|\textbf{x}\textbf{y}}$ denotes a normalised state (note that $\rho_{\textbf{a}\textbf{b}|\textbf{x}\textbf{y}}$ may be arbitrary if $p_{\textbf{a}\textbf{b}|\textbf{x}\textbf{y}}=0$). Observe that due to (\ref{0con1}-\ref{0con3}), $\left\{p_{\textbf{a}\textbf{b}|\textbf{x}\textbf{y}}\right\}_{\textbf{a}\textbf{b}|\textbf{x}\textbf{y}}$ forms a correlation box that obeys constraints of time-ordering no-signaling (TONS) - see for example \cite{BPA}.

Fix any $\textbf{x}_{\overline{k}},\textbf{y}_{\overline{k}}$. Observe that according to (\ref{0con2},\ref{0con3}) the following no-signaling assemblage
\begin{equation}\label{cond1}
\Sigma^{(1)}_{\textbf{x}_{\overline{k}}\textbf{y}_{\overline{k}}}:=\left\{\sum_{\textbf{a}_{\overline{k}},\textbf{b}_{\overline{k}}}\mathrm{Tr}_{\textbf{i}_{\overline{k}}}(\sigma_{a_1b_1\ldots a_{n}b_{n}|x_1y_1\ldots x_{n}y_{n}})\right\}_{a_k,b_k,x_k,y_k}
\end{equation}does not depend on the choice of those $x_s$ and $y_r$ for which $s,r>k$ (in other words one can define $\Sigma^{(1)}_{\textbf{x}_{<k}\textbf{y}_{<k}}:=\Sigma^{(1)}_{\textbf{x}_{\overline{k}}\textbf{y}_{\overline{k}}}$). I particular case when $\Sigma^{(1)}_{\textbf{x}_{<k}\textbf{y}_{<k}}$ does not depend on the choice of any $x_i,y_j$ (for $i,j\neq k$) we will simply write $\Sigma^{(1)}_{k}:=\Sigma^{(1)}_{\textbf{x}_{< k}\textbf{y}_{<k}}$ for unambiguous assemblage related to time $t_k$.

Note that for any fixed $\text{a}_{< n},\text{b}_{< n},\text{x}_{< n},\text{y}_{< n}$ we have (due to conditions (\ref{0con2},\ref{0con3}))
\begin{equation}\label{cond2}
\sum_{a_n}\mathrm{Tr}_{1,\ldots, n-1}(\sigma_{a_1b_1\ldots a_{n}b_{n}|x_1y_1 \ldots x_{n}y_{n}})=\sum_{a_n}\mathrm{Tr}_{1,\ldots, n-1}(\sigma_{a_1b_1\ldots a_{n}b_{n}|x_1y_1\ldots x'_{n}y_{n}})
\end{equation}and
\begin{equation}\label{cond3}
\sum_{b_n}\mathrm{Tr}_{1,\ldots, n-1}(\sigma_{a_1b_1\ldots a_{n}b_{n}|x_1y_1 \ldots x_{n}y_{n}})=\sum_{b_n}\mathrm{Tr}_{1,\ldots, n-1}(\sigma_{a_1b_1\ldots a_{n}b_{n}|x_1y_1\ldots x_{n}y'_{n}})
\end{equation}for any choice of $x_n,x'_n,y_n,y'_n$. If so then
\begin{equation}\label{subass}
\Sigma^{(1)}_{\textbf{a}_{< n}\textbf{b}_{< n}\textbf{x}_{< n}\textbf{y}_{< n}}:=\left\{\mathrm{Tr}_{1,\ldots, n-1}(\sigma_{a_1b_1\ldots a_{n}b_{n}|x_1y_1\ldots x_{n}y_{n}})\right\}_{a_n,b_n,x_n,y_n}
\end{equation}is a no-signaling assemblage up to a positive rescaling (it is a subnormalised assemblage).

From now on let $a_i,b_i,x_i,y_i\in \left\{0,1\right\}$ for any $i=1,\ldots, n$. In what follows, when this not cause a confusion we will omit subscripts while referring to sets.

\subsection{The steering inequality used in the 1sDI Protocol II} 
In the 1sDI protocol II in Fig. \ref{protocol1sDI} from the main text, we consider the simplest nontrivial case of tripartite steering with two uncharacterised subsystems, i.e., with $a,b,x,y\in \left\{0,1\right\}$. Note that a no-signaling assemblage can be then seen as a box of positive operators (i.e. subnormalised states) where $(a|x)$ label rows and $(b|y)$ label columns, i.e. 
\begin{equation}\label{eq:NS-assemblage}
\Sigma^{(C)}=\begin{pmatrix}
\begin{array}{cc|cc}
 \sigma^{(C)}_{00|00} &  \sigma^{(C)}_{01|00} & \sigma^{(C)}_{00|01} &  \sigma^{(C)}_{01|01} \\  
 \sigma^{(C)}_{10|00} & \sigma^{(C)}_{11|00}& \sigma^{(C)}_{10|01}  & \sigma^{(C)}_{11|01} \\ \hline
 \sigma^{(C)}_{00|10} & \sigma^{(C)}_{01|10} & \sigma^{(C)}_{00|11}&  \sigma^{(C)}_{01|11}  \\
   \sigma^{(C)}_{10|10} & \sigma^{(C)}_{11|10} & \sigma^{(C)}_{10|11} & \sigma^{(C)}_{11|11}
\end{array}
\end{pmatrix}.
\end{equation}In particular, LHS assemblages are convex combinations of extremal boxes  (of operators) which have only four nonzero positions occupied by the same pure state forming a rectangle with exactly one element for each pair $(x,y)$.

We also recall the notions of {\it similarity} and  {\it inflexibility} of no-signaling assemblages with pure rank one elements $\sigma_{ab|xy}^{(C)}$ introduced in \cite{RBRH20}.

\begin{definition}\label{similarity}
Consider a general no-signaling assemblage $\Sigma^{(C)}$ as in Eq.(\ref{eq:NS-assemblage}) with all positions occupied by at most rank one operators and denote it by $\Sigma^{(C)}=\left\{p_i|\psi^{(C)}_i\rangle \langle \psi^{(C)}_i|\right\}_i$, where $i=(ab|xy)$ and $p_i=\mathrm{Tr}(\sigma^{(C)}_i)$.  Any other assemblage $\tilde{\Sigma}^{(C)}=\left\{q_i|\psi^{(C)}_i\rangle \langle \psi^{(C)}_i|\right\}_i$ with the same states at the same positions and the additional property that $p_i=0$ implies $q_i=0$ is called {\it similar} to $\Sigma^{(C)}$.
\end{definition}
Note that the above relation is not symmetric, i.e. it may happen that $\tilde{\Sigma}^{(C)}$ is similar to $\Sigma^{(C)}$, 
but $\Sigma^{(C)}$ is not similar to $\tilde{\Sigma}^{(C)}$. The second concept is defined here

\begin{definition}\label{inflexibility}
An assemblage $\Sigma^{(C)}$ is called \textit{inflexible} if for any $\tilde{\Sigma}^{(C)}$ similar to $\Sigma^{(C)}$ we 
get $\Sigma^{(C)}=\tilde{\Sigma}^{(C)}$.
\end{definition}

Note that in particular {\it inflexibility implies extremality} in the set of all no-signaling assemblages. In the two-party steering scenario, it is known that all no-signaling assemblages admit quantum realisation \cite{HJW93}. As such, it is easy to identify extremal no-signaling assemblages (with quantum realisations) that mandate a decoupling from any adversarial system. On the other hand, none of the extremal points so far have been found to give rise to a perfect random bit under the constraint of arbitrarily weak measurement independence. Such a random bit is more readily obtained from a three-party test on a GHZ state for arbitrary initial $\epsilon$. This fact, coupled with a recent discovery that there are extremal non-trivial non-signaling assemblages in the tripartite scenario that admit quantum realisation motivates us to formulate the 1sDI protocol with three honest users. Specifically the following proposition was shown in \cite{RBRH20}. 

\begin{proposition}(\cite{RBRH20})\label{proposition_genuine}
For any pure genuine three-party entangled state $|\psi^{(ABC)}\rangle\in \mathbb{C}^{2}\otimes \mathbb{C}^{2}\otimes \mathbb{C}^{d}$ there exists a pair of PVMs with two outcomes on subsystems A and B respectively such that a no-signaling assemblage $\Sigma^{(C)}$ obtained by these measurements is inflexible, extremal, not LHS, and uniquely and maximally violates some steering inequality $F$.
\end{proposition}
Note that in contrast to the above Proposition, it has been shown in the fully device-independent setting that no non-trivial extremal no-signaling boxes admit quantum realization \cite{RTHHPRL}. In the 1sDI protocol, we use a specific instance of  Prop. \ref{proposition_genuine} in which the steering inequality $F$ allows for randomness extraction for arbitrarily measurement independence, i.e., when the measurement settings are chosen using SV sources with arbitrary $\epsilon < 1/2$. Furthermore, the maximal violation of the chosen steering inequality $F$ is self-testing for the GHZ state and certifies that the output bit of one party (Alice) is perfectly random. 

In particular, we consider the steering inequality defined as follows. Consider the three-qubit GHZ state $|\psi^{(ABC)}\rangle=\frac{1}{\sqrt{2}}\left(|000\rangle +|111\rangle\right)$. Let the no-signaling assemblage $\Sigma^{(C)}_{GHZ}$ be given by $\sigma^{(C)}_{ab|xy}=\mathrm{Tr}_{AB}(P^{(A)}_{a|x}\otimes Q^{(B)}_{b|y}\otimes \mathds{1}|\psi^{(ABC)}\rangle \langle \psi^{(ABC)}|)$ with $P^{(A)}_{0|0}=Q^{(B)}_{0|0}=|+\rangle \langle +|$ and $P^{(A)}_{0|1}=Q^{(B)}_{0|1}=|0\rangle \langle 0|$, i.e.
\begin{equation}\nonumber
\Sigma^{(C)}_{GHZ}=\frac{1}{4}\begin{pmatrix}
\begin{array}{cc|cc}
 |+\rangle \langle +| &  |-\rangle \langle -| & |0\rangle \langle 0| &  |1\rangle \langle 1| \\  
 |-\rangle \langle -| & |+\rangle \langle +|& |0\rangle \langle 0|  &|1\rangle \langle 1|\\ \hline
 |0\rangle \langle 0| & |0\rangle \langle 0| &2|0\rangle \langle 0|&  0  \\
   |1\rangle \langle 1| &  |1\rangle \langle 1|& 0 & 2|1\rangle \langle 1| 
\end{array}
\end{pmatrix}.
\end{equation}
We first consider the steering functional from \cite{RBRH20} given by
\begin{equation}\label{expression}
F_{\Sigma^{(C)}_{GHZ}}(\tilde{\Sigma}^{(C)})=\sum_{a,b,x,y = 0,1}\mathrm{Tr}(\rho_{ab|xy}\tilde{\sigma}^{(C)}_{ab|xy}),
\end{equation}
where 
\begin{equation}
 \rho_{ab|xy}=
\begin{cases}
0\ \ \ \mathrm{for}\ \sigma^{(C)}_{ab|xy}=0,\\
\frac{\sigma^{(C)}_{ab|xy}}{\mathrm{Tr}(\sigma^{(C)}_{ab|xy})} \ \ \  \mathrm{for}\ \sigma^{(C)}_{ab|xy}\neq 0.
\end{cases}
\end{equation}
Observe that $F_{\Sigma^{(C)}_{GHZ}}(\tilde{\Sigma}^{(C)})\leq 4$ for all no-signalling assemblages and as shown in \cite{RBRH20}, the maximal value of $F_{\Sigma^{(C)}_{GHZ}}(\tilde{\Sigma}^{(C)})$ is uniquely obtained for $\Sigma^{(C)}_{GHZ}$. It was also shown in \cite{RBRH20} that the maximum value of the steering functional over all local hidden state assemblages is given by $
c=\sup_{|\phi\rangle}\mathrm{Tr}\left[(3|0\rangle \langle 0|+|+\rangle \langle +|)|\phi\rangle \langle \phi|\right]=\frac{4+\sqrt{10}}{2} < 4$. 

The parties in the Protocol II test for the complementary steering functional given as $4 - F_{\Sigma^{(C)}_{GHZ}}(\tilde{\Sigma}^{(C)})$ which achieves its minimum algebraic value of $0$ in quantum theory while the LHS minimum is $\frac{4-\sqrt{10}}{2}$. In other words, the parties check that the quantity 
\begin{eqnarray}
Z_{\text{S}} := \frac{1}{M} \sum_{i=1}^{M} B_{\text{S}}(a_i, b_i, x_i, y_i) = 0
\end{eqnarray} 
where
\begin{eqnarray}
B_{\text{S}}(a_i, b_i, x_i, y_i) = \left(\mathds{1} - \rho_{a_i,b_i|x_i,y_i} \right) \cdot \tilde{\sigma}^{(C)}_{a_i,b_i|x_i,y_i}
\end{eqnarray}
and $\tilde{\sigma}^{(C)}_{a_i,b_i|x_i,y_i}$ is the state of Charlie's system for the output-input pairs $(a_i,x_i)$ and $(b_i,y_i)$ on Alice and Bob's systems. Given that the algebraic violation of the inequality is achieved by quantum correlations, the steering test works for arbitrary initial $\epsilon$ \cite{PRB+14}. 
Finally, it is also clear that Alice's measurements $P_{0|0}^{(A)}$ and $P_{0|1}^{(A)}$ on the GHZ state give rise to perfectly random outcomes.  
In the next subsection, we show that the properties of self-testing, extremality and perfect randomness in the measurement output of one party persist in the time-ordered no-signaling steering scenario where the parties perform sequential measurements on their respective subsystems. 

\subsection{Relative inflexibility of sequential assemblages}

\begin{thm}\label{thm} Let $\Sigma^{(n)}$ be a sequential assemblage such that $\Sigma^{(1)}_i=\Sigma_{i}$ for any $i=1,\ldots, n$ and each $\Sigma_{i}$ is an inflexible assemblage, then automatically $\Sigma^{(n)}=\otimes_i^n \Sigma_{i}$.
\end{thm}

\begin{proof}

For simplicity we will use the notation in which $|\psi\rangle \langle \psi|=\Psi$ for any pure state. Consider first the special case with $n=2$ and $\Sigma^{(2)}$ such that
\begin{equation}\label{eq}
\Sigma^{(1)}_1=\Sigma_{1}:=\left\{q^{(1)}_{ab|xy}\Psi^{(1)}_{ab|xy}\right\}
\end{equation}and
\begin{equation}\label{eq2}
\Sigma^{(1)}_2=\Sigma_{2}:=\left\{q^{(2)}_{ab|xy}\Psi^{(2)}_{ab|xy}\right\}
\end{equation}where $\Sigma_{1},\Sigma_{2}$ stand for some fixed inflexible assemblages. We will show that $\Sigma^{(2)}=\Sigma^{(1)}_1\otimes\Sigma^{(1)}_2$.

Observe that 
\begin{equation}
p_{a_1b_1a_2b_2|x_1y_1x_2y_2}=p_{a_2b_2|a_1b_1x_1y_1x_2y_2}p_{a_1b_1|x_1y_1x_2y_2}
\end{equation}and
\begin{equation}
p_{a_1b_1|x_1y_1x_2y_2}=p_{a_1b_1|x_1y_1}
\end{equation} by no-signaling (backward in time, from $t_2$ to $t_1$ - see conditions (\ref{0con2},\ref{0con3})). Consider now the spectral decomposition given by
\begin{equation}
\rho_{a_1b_1a_2b_2|x_1y_1x_2y_2}=\sum_j \lambda^{(j)}_{a_1b_1a_2b_2|x_1y_1x_2y_2}\Phi^{(j)}_{a_1b_1a_2b_2|x_1y_1x_2y_2}.
\end{equation}By equality (\ref{eq}) for $\Sigma^{(1)}_1$ we get
\begin{equation}\label{con1}
q^{(1)}_{a_1b_1|x_1y_1}\Psi^{(1)}_{a_1b_1|x_1y_1}=\sum_{a_2,b_2}p_{a_2b_2|a_1b_1x_1y_1x_2y_2}p_{a_1b_1|x_1y_1}\mathrm{Tr}_{2}(\rho_{a_1b_1a_2b_2|x_1y_1x_2y_2})
\end{equation} for any fixed $a_1,b_1,x_1,y_1$ and any choice of $x_2,y_2$. Note that $q^{(1)}_{a_1b_1|x_1y_1}=0$ if and only if $p_{a_1b_1|x_1y_1}=0$. If this is not the case, (\ref{con1}) may be true if and only if each $\Phi^{(j)}_{a_1b_1a_2b_2|x_1y_1x_2y_2}$ (such that $p_{a_2b_2|a_1b_1x_1y_1x_2y_2}\neq 0$) is related to the product state
\begin{equation}\label{phi}
\Phi^{(j)}_{a_1b_1a_2b_2|x_1y_1x_2y_2}=\Psi^{(1)}_{a_1b_1|x_1y_1}\otimes \tilde{\Psi}^{(j)}_{a_1b_1a_2b_2|x_1y_1x_2y_2}
\end{equation}and $p_{a_1b_1|x_1y_1}=q^{(1)}_{a_1b_1|x_1y_1}$. On the other hand, by equality (\ref{eq2}), we get
\begin{equation}\label{con2}
q^{(2)}_{a_2b_2|x_2y_2}\Psi^{(2)}_{a_2b_2|x_2y_2}=\sum_{a_1,b_1}p_{a_2b_2|a_1b_1x_1y_1x_2y_2}p_{a_1b_1|x_1y_1}\mathrm{Tr}_{1}(\rho_{a_1b_1a_2b_2|x_1y_1x_2y_2})
\end{equation} for any fixed $a_2,b_2,x_2,y_2$ and any choice of $x_1,y_1$. Note that $q^{(2)}_{a_2b_2|x_2y_2}=0$ implies 
$p_{a_1b_1|x_1y_1}=0$ or $p_{a_2 b_2|a_1b_1x_1y_1x_2y_2}=0$ (for all $a_1,b_1$) and if $p_{a_2b_2|a_1b_1x_1y_1x_2y_2}p_{a_1b_1|x_1y_1}\neq 0$ then for each $j$ (by (\ref{phi})) 
\begin{equation}
\Phi^{(j)}_{a_1b_1a_2b_2|x_1y_1x_2y_2}=\Psi^{(1)}_{a_1b_1|x_1y_1}\otimes \Psi^{(2)}_{a_2b_2|x_2y_2}.
\end{equation}From this without loss of generality one can write
\begin{equation}
\Sigma^{(2)}=\left\{p_{a_1b_1|x_1y_1}p_{a_2b_2|a_1b_1, x_1y_1x_2y_2}\Psi^{(1)}_{a_1b_1|x_1y_1}\otimes \Psi^{(2)}_{a_2b_2|x_2y_2}\right\}.
\end{equation}For some fixed $a_1,b_1,x_1,y_1$ (such that $p_{a_1b_1|x_1y_1}\neq 0$) consider $\Sigma^{(1)}_{a_1b_1x_1y_1}=\left\{\mathrm{Tr}_{1}\left(\sigma_{a_1b_1a_2b_2|x_1y_1x_2y_2}\right)\right\}$ (compare with (\ref{subass})). This object (up to normalisation) has to be a no-signaling assemblage similar to $\Sigma_{2}$, but from inflexibility of $\Sigma_{2}$ we have
\begin{equation}
p_{a_1b_1|x_1y_1}p_{a_2b_2|a_1b_1x_1y_1x_2y_2}=\alpha_{a_1b_1|x_1y_1}q^{(2)}_{a_2b_2|x_2y_2}
\end{equation}for some positive constant $\alpha_{a_1b_1|x_1y_1}$. Therefore 
\begin{equation}
p_{a_1b_1|x_1y_1}=\sum_{a_2,b_2}p_{a_1b_1|x_1y_1}p_{a_2b_2|a_1b_1x_1y_1x_2y_2}=\sum_{a_2,b_2}\alpha_{a_1b_1|x_1y_1}q^{(2)}_{a_2b_2|x_2y_2}
\end{equation}and finally $\alpha_{a_1b_1|x_1y_1}=p_{a_1b_1|x_1y_1}$. Since in has been already established that $p_{a_1b_1|x_1y_1}=q^{(1)}_{a_1b_1|x_1y_1}$, in the end we see that
\begin{equation}
\Sigma^{(2)}=\left\{q^{(1)}_{a_1b_1|x_1y_1}q^{(2)}_{a_2b_2|x_2y_2}\Psi^{(1)}_{a_1b_1|x_1y_1}\otimes \Psi^{(2)}_{a_2b_2|x_2y_2}\right\}=\Sigma_1\otimes \Sigma_2.
\end{equation}

To conclude the proof we need to show that if the statement is true for a given $n$, it must be true as well for $n+1$. In order to do that consider sequential assemblage $\Sigma^{(n+1)}$ such that $\Sigma^{(1)}_i=\Sigma_{i}$ for any $i=1,\ldots, n+1$ and each $\Sigma_{i}$ is inflexible. Define new sequential assemblage by discarding part related to time $t_{n+1}$, i.e. define
\begin{equation}
\tilde{\Sigma}^{(n)}=\left\{\sum_{a_{n+1},b_{n+1}}\mathrm{Tr}_{n+1}\left(\sigma_{a_1b_1\ldots a_{n+1}b_{n+1}|x_1y_1\ldots x_{n+1}y_{n+1}}\right)\right\}.
\end{equation}Because for any $i=1,\ldots ,n$ we have $\tilde{\Sigma}^{(1)}_i=\Sigma^{(1)}_i=\Sigma_{i}$, due to inductive assumption we obtain
\begin{equation}\label{prod}
\tilde{\Sigma}^{(n)}=\otimes_{i}^n\Sigma_{i}.
\end{equation}
Observe that 
\begin{equation}
p_{a_1b_1\ldots a_{n+1}b_{n+1}|x_1y_1\ldots x_{n+1}y_{n+1}}=p_{a_{n+1}b_{n+1}|a_1b_2\ldots a_nb_nx_1y_1\ldots x_{n+1}y_{n+1}}p_{a_1b_1\ldots a_nb_n|x_1y_1\ldots x_{n+1}y_{n+1}}
\end{equation}and
\begin{equation}
p_{a_1b_1\ldots a_nb_n|x_1y_1\ldots x_{n+1}y_{n+1}}=p_{a_1b_1\ldots a_nb_n|x_1y_1\ldots x_{n}y_{n}}
\end{equation} by no-signaling (backward in time). 

Moreover, because of (\ref{prod}) and $\Sigma^{(1)}_{n+1}=\Sigma_{n+1}$, by the similar arguments as considered above in specific case $n=2$, one can see that
\begin{equation}
p_{a_1b_1\ldots a_nb_n|x_1y_1\ldots x_{n}y_{n}}=\prod_i^n q^{(i)}_{a_ib_i|x_iy_i}
\end{equation} and without loss of generality $\Sigma^{(n+1)}$ can be expressed as
\begin{equation}\nonumber
\Sigma^{(n+1)}=\left\{\prod_i^n q^{(i)}_{a_ib_i|x_iy_i}p_{a_{n+1}b_{n+1}|a_1b_2\ldots a_nb_nx_1y_1\ldots x_{n+1}y_{n+1}}\otimes_i^n \Psi^{(i)}_{a_ib_i|x_iy_i}\otimes \Psi^{(n+1)}_{a_{n+1}b_{n+1}|x_{n+1}y_{n+1}}\right\}.
\end{equation}For a fixed choice of $a_1,b_2,\ldots, a_n,b_nx_1,y_1,\ldots ,x_{n},y_{n}$ such that $\prod_i^n q^{(i)}_{a_ib_i|x_iy_i}\neq 0$, one may consider subnormalised assemblage $\Sigma^{(1)}_{\textbf{a}_{\leq n}\textbf{b}_{\leq n}\textbf{x}_{\leq n}\textbf{y}_{\leq n}}=\left\{\mathrm{Tr}_{1,\ldots, n}\left(\sigma_{a_1b_1\ldots a_{n+1}b_{n+1} |x_1y_1\ldots x_{n+1}y_{n+1})}\right)\right\}$ and once again by the same arguments like above, one can show that
\begin{equation}
p_{a_{n+1}b_{n+1}|a_1b_2\ldots a_nb_nx_1y_1\ldots x_{n+1}y_{n+1}}=q^{(n+1)}_{a_{n+1}b_{n+1}|x_{n+1}y_{n+1}}
\end{equation} and finally 
\begin{equation}\nonumber
\Sigma^{(n+1)}=\left\{\prod_i^{n+1} q^{(i)}_{a_ib_i|x_iy_i}\otimes_{i=1}^{n+1} \Psi^{(i)}_{a_ib_i|x_iy_i}\right\}=\otimes_{i}^{n+1}\Sigma_{i}.
\end{equation}By the rule of inductive procedure the proof is completed.
\end{proof}

In particular if for any $i=1,\ldots, n$ we have $\Sigma^{(1)}_i=\Sigma_{fix}$ for some inflexible assemblage $\Sigma_{fix}$, then $\Sigma^{(n)}=\otimes_i^n \Sigma_{fix}$

\subsection{Self-testing result for sequential assemblages}

For any $i=1,\ldots, n$ consider a given inflexible assemblage $\Sigma_{i}=\left\{\sigma^{(i)}_{ab|xy}\right\}=\left\{q^{(i)}_{ab|xy}\Psi^{(i)}_{ab|xy}\right\}$ where once more $\Psi$ stands for pure state $|\psi\rangle \langle \psi|$. We know \cite{RBRH20} that it can be self-tested by some functional (steering inequity) $F_i$ defined by
\begin{equation}
 F_i(\Sigma)=\sum_{a,b,x,y = 0,1}\mathrm{Tr}(\tau^{(i)}_{ab|xy}\sigma_{ab|xy}), \ \ \ \ \ \tau^{(i)}_{ab|xy}=
\begin{cases}
0\ \ \ \mathrm{for}\ \sigma^{(i)}_{ab|xy}=0,\\
\frac{\sigma^{(i)}_{ab|xy}}{\mathrm{Tr}(\sigma^{(i)}_{ab|xy})} \ \ \  \mathrm{for}\ \sigma^{(i)}_{ab|xy}\neq 0.
\end{cases}
\end{equation}where $\Sigma=\left\{\sigma_{ab|xy}\right\}$ stands for any no-signaling assemblage ($F_i(\Sigma)=4$ if and only if $\Sigma=\Sigma_i$ - see \cite{RBRH20}).

Now, based on the structure of each $F_i$, define a new functional $F_{1\ldots n}^{(n)}$ acting on sequential assemblages (described in the Section \ref{sec1}) by
\begin{equation}
 F_{1,\ldots, n}^{(n)}(\Sigma^{(n)}):=\sum_{\textbf{a},\textbf{b},\textbf{x},\textbf{y}}\mathrm{Tr}(\tau_{\textbf{a}\textbf{b}|\textbf{x}\textbf{y}}\sigma_{\textbf{a}\textbf{b}|\textbf{x}\textbf{y}}), \ \ \ \ \ \tau_{\textbf{a}\textbf{b}|\textbf{x}\textbf{y}}=\otimes_i^{n}\tau^{(i)}_{a_ib_i|x_iy_i}
\end{equation}where $\Sigma^{(n)}=\left\{\sigma_{\textbf{a}\textbf{b}|\textbf{x}\textbf{y}}\right\}$ stands for any sequential no-signaling assemblage on $n$ time points. Observe that the maximal value for $F_{1\ldots n}^{(n)}(\tilde{\Sigma}^{(n)})$ if given by $4^n$ and if so, then we can write (for this sequential assemblage $\tilde{\Sigma}^{(n)}$ which maximizes $F_{1\ldots n}^{(n)}$)
\begin{equation}
\tilde{\Sigma}^{(n)}=\left\{p_{\textbf{a}\textbf{b}|\textbf{x}\textbf{y}}\rho_{\textbf{a}\textbf{b}|\textbf{x}\textbf{y}}\right\}
\end{equation}with $\rho_{\textbf{a}\textbf{b}|\textbf{x}\textbf{y}}=\prod_i^{n}\Psi^{(i)}_{a_ib_i|x_iy_i}$ and $p_{\textbf{a}\textbf{b}|\textbf{x}\textbf{y}}=0$ whenever there exist $i$ such that $q^{(i)}_{a_ib_i|x_iy_i}=0$. 

Fix any $i$. According to the previous observation we have (compare with formula (\ref{cond1}))
\begin{equation}
F_{1,\ldots, n}^{(n)}(\tilde{\Sigma}^{(n)})=\sum_{\textbf{x}_{\overline{i}},\textbf{y}_{\overline{i}}}\left(\sum_{a_i,b_i,x_iy_i}\mathrm{Tr}\left((\tilde{\Sigma}^{(1)}_{\textbf{x}_{\overline{i}},\textbf{y}_{\overline{i}}})_{a_ib_i|x_iy_i}\right)\right)=\sum_{\textbf{x}_{\overline{i}},\textbf{y}_{\overline{i}}}F_i(\tilde{\Sigma}^{(1)}_{\textbf{x}_{\overline{i}},\textbf{y}_{\overline{i}}})=4^n. 
\end{equation}Note that as above expression boils down to the sum of $4^{n-1}$ terms and that each term is given by a value of functional $F_i$ on some no-signaling assemblage $\tilde{\Sigma}^{(1)}_{\textbf{x}_{\overline{i}},\textbf{y}_{\overline{i}}}$. Therefore, from this equality we get $F_i(\tilde{\Sigma}^{(1)}_{\textbf{x}_{\overline{i}},\textbf{y}_{\overline{i}}})=4$ for each choice $\textbf{x}_{\overline{i}},\textbf{y}_{\overline{i}}$ and by uniqueness $\tilde{\Sigma}^{(1)}_{\textbf{x}_{\overline{i}},\textbf{y}_{\overline{i}}}=\Sigma_{i}^{(1)}=\Sigma_{i}$. As the choice of $i$ was arbitrary, $\tilde{\Sigma}^{(1)}_{i}=\Sigma_{i}$ for any $i$. We can see now that $\tilde{\Sigma}^{(n)}$ fulfills assumption of Theorem \ref{thm} and we derive at the following results.

\begin{thm}\label{thm2} $F^{(n)}_{1,\ldots, n}(\Sigma^{(n)})=4^n$ if and only if $\Sigma^{(n)}=\otimes_i^n \Sigma_{i}$, where each inflexible assemblage $\Sigma_{i}$ maximizes functional $F_i$.
\end{thm}

Note that in particular if $F_i=F$ for every $i=1,\dots, n$ and $F$ is a functional maximized by some fixed inflexible assemblage $\Sigma_{fix}$, then $F^{(n)}(\Sigma^{(n)})=4^n$ if and only if $\Sigma^{(n)}=\otimes_i^n \Sigma_{fix}$.

\end{document}